%% file: draft.tex
\documentclass[12pt, letterpaper]{article}

\input{preamble}

\title{Dynamic User Competition and Miner Behavior\\ in the Bitcoin Market\thanks{We thank Eric Budish, Shin'ichiro Matsuo, Nozomu Muto, Tatsuaki Okamoto, Jesse Perla, Go Yamamoto, and seminar participants at 
The Foundation for the Promotion of Industrial Science at the University of Tokyo, Kyoto University, NTT Research Upgrade 2021 at NTT Research, and the Digital Currency and Finance Workshop at the University of Tokyo 
for their helpful comments. 
Part of this research was conducted while Kamada was affiliated with NTT Research. Kamada thanks for their hospitality. Yuji Tamakoshi provided excellent research assistance.}}

\date{\today}

\author{Yuichiro Kamada\footnote{Haas School of Business, University of California Berkeley, 2220 Piedmont Avenue, Berkeley, CA
94720-1900, USA, and Graduate School of Economics, University of Tokyo, 7-3-1 Hongo, Bunkyo-ku, Tokyo,
113-0033, Japan. E-mail: \href{mailto:y.cam.24@gmail.com}{y.cam.24@gmail.com}} \and Shunya Noda\footnote{Graduate School of Economics, University of Tokyo, 7-3-1 Hongo, Bunkyo-ku, Tokyo, 113-0033, Japan. E-mail: \href{mailto:shunya.noda@e.u-tokyo.ac.jp}{shunya.noda@e.u-tokyo.ac.jp}}}

\begin{document}

\onehalfspacing

\maketitle

\begin{abstract}
We develop a dynamic model of the Bitcoin market where users set fees themselves and miners decide whether to operate and whom to validate based on those fees. Our analysis reveals how, in equilibrium, users adjust their bids in response to short-term congestion (i.e., the amount of pending transactions), how miners decide when to start operating based on the level of congestion, and how the interplay between these two factors shapes the overall market dynamics. The miners hold off operating when the congestion is mild, which harms social welfare. However, we show that a block reward (a fixed reward paid to miners upon a block production) can mitigate these inefficiencies. We characterize the socially optimal block reward and demonstrate that it is always positive, suggesting that Bitcoin's halving schedule may be suboptimal.
\end{abstract}

\JEL{D44, G20, L86}

\keywords{Bitcoin, Cryptocurrency, Blockchain Economics, Transaction Fee Bidding, Decentralized Market Design, Pay-as-Bid Auction}

\newpage

\section{Introduction}

Markets in which consumers set fees themselves, and suppliers decide whom to serve based on those fees, are becoming increasingly common. Examples range from shipping and freight marketplaces to freelancing platforms, cloud computing spot markets, and online tutoring. These markets, which we call the \textit{consumer-driven fee-setting markets}, share a fundamental challenge: both consumers and suppliers participate in a decentralized manner, leading to short-term fluctuations in supply and demand. This volatility can cause market congestion, influencing equilibrium fees set by consumers and entry decisions by suppliers. Understanding how such consumer-driven fee-setting markets function dynamically is crucial for ensuring efficiency and robustness.

The Bitcoin market (and markets for other cryptocurrencies) is a prominent example of a consumer-driven fee-setting market that highlights these challenges. In the Bitcoin market, users (consumers) specify the transaction fees they are willing to pay for validation, while ``miners'' select which transactions to validate. For a miner to be able to validate transactions, they expend computational resources to win the right to collect a certain number of transaction requests, validate them, put them in a ``block,'' and receive two forms of compensation: the fees paid by users and a fixed ``block reward'' paid with newly issued coins. The decentralized nature of Bitcoin means that the congestion level can fluctuate rapidly, affecting both the user bidding behavior and miner participation. This interplay between the transaction fee competition and miner incentives makes Bitcoin a particularly insightful case for studying decentralized markets with dynamic supply-demand imbalances.

Prior research, such as \citet{easley2019mining,huberman2021monopoly,hinzen2022bitcoin,john2025proof}, studied models of the Bitcoin market in a static setting in the sense that they assume that users bid their fees without observing the current congestion level, and the miners exert a constant mining effort. In reality, however, both users' and miners' decision problems are dynamic: the users can easily observe the current congestion level and the wallet apps assist them in choosing their bid accordingly, and the miners can also observe the congestion level and condition their operation decision on it. We model such dynamics,  which enables us to understand how the dynamic nature of the market affects the behavior of the users and miners as well as the market outcome. Our analysis shows that the insights from the static models carry over to the dynamic setting, while also revealing novel aspects that pertain to the dynamics.\footnote{In contrast, the prior research introduces heterogeneity in users' utility from transaction validation and their disutility from waiting time. We do not introduce such heterogeneity in order to be able to analyze a complex dynamic model.}

In this paper, we develop a dynamic model of user and miner behavior that captures the short-run dynamics of the Bitcoin market. Our analysis reveals that users adjust their fee bids dynamically based on the system throughput, the short-term transaction congestion, and the miners' operation policy. Additionally, we show that when few transactions are pending, miners may temporarily suspend operations due to cost inefficiencies, which can in turn lead to a surge in transaction fees. Finally, we highlight the role of block rewards, showing that sustaining a positive block reward enhances social welfare by mitigating miner suspension and ensuring more stable transaction processing.

In our model, users enter the market continuously. Each user observes the distribution of bids from users who arrived earlier, anticipates the bids of future users, and decides their own bids optimally.
On the other side of the market, the miners decide whether to operate to produce a block. Blocks arrive probabilistically while the miners are operating, and hence it is possible that a block does not arrive for a long time. When a block arrival is delayed as a realization of such an arrival process, a large number of users compete for the limited space, prompting users to submit higher bids. We construct a simple and tractable model describing this inherently complex decision-making process. 

To separate the effects of the two sides' decision-making problems, we first consider a simplified model where the miners always operate. We characterize the unique equilibrium---specifying a bid as a function of the number of pending transactions---of such a model. We analyze how such ``bid function'' varies when the parameters such as the system throughput change.

We then present a full model that endogenizes the miners' decision-making. The number of users in the market and the magnitude of their bids influence the miners' incentives to operate in order to earn transaction fees. When the number of waiting users is low, their bids are low as well, weakening the miners' incentives to bear the costs of operation. As a result, in equilibrium, the miners suspend operations until the number of pending transaction requests becomes sufficiently large. This temporal suspension reduces the block supply, thereby intensifying user competition, which in turn leads to higher bids. To our knowledge, this is the first study to analyze this interplay.

We demonstrate that such temporary suspensions by miners reduce social welfare. Furthermore, we show that this issue can be addressed by providing miners with block rewards, which are paid regardless of the amount of pending transactions. Specifically, we characterize the block reward level that maximizes social welfare and show that this level is always positive. This occurs because miners' operational decisions are based solely on their private profits while maximizing social welfare requires balancing the combined benefits of users and miners. Block rewards act as subsidies to bridge this gap. Not only are block rewards essential for maximizing social welfare, but they also positively affect user surplus. The result sheds light on the novel role that the block reward plays and implies the suboptimality of Bitcoin’s policy, where it schedules block rewards to halve every four years and eventually reach zero.

\subsection{Institutional Detail}

In this section, we briefly introduce the institutions of Bitcoin that are directly relevant to our study. For more comprehensive institutional details, particularly regarding the design intentions of the system, we refer the reader to \citet{narayanan2016bitcoin}.

Bitcoin, established by \citet{nakamoto2008bitcoin}, is the oldest and largest cryptocurrency, where both users and \emph{miners} (record-keepers) participate in a decentralized manner. Many cryptocurrencies, including Bitcoin, manage the history of validated transactions using a ledger called \emph{blockchain}. A collection of transactions up to a certain limit (determined by data size) is called a \emph{block}. Miners collect users' transaction requests, generate blocks, and add them to the blockchain, thereby validating new transactions.

A user wishing to initiate a transaction creates a request to validate the transaction, following a specified format, and broadcasts it to miners worldwide. At this stage, the transaction is not yet considered to be validated; rather, it is placed in a virtual database of pending transactions known as the \emph{mempool}. Miners independently verify the validity of pending transactions (e.g., checking proper signatures and sufficient balances) and, if deemed valid, include them in a block they generate and attempt to append it to the blockchain. A transaction is considered to be validated once it is included in a block appended to the blockchain.

To ensure security, the block production is subject to specific requirements. In the Bitcoin system, a miner who wins a lottery, which is generated using a cryptographic technology, is selected to append the next block.\footnote{Specifically, what is called a ``cryptographic hash function'' is utilized for generating the lottery. This method is called ``Proof-of-Work'' and is adopted by many cryptocurrencies besides Bitcoin.} The lottery winning probability, automatically determined by an algorithm, is updated only once every two weeks, making it constant in the short term.\footnote{In the patient-user model studied in Appendix~\ref{sec: appendix patient users}, the block arrival rate is assumed to be independent of the timing of the past block arrivals and it is constant over time. This assumption is justified by the fact that Bitcoin adjusts the lottery winning probability only occasionally. In contrast, some newer cryptocurrencies adjust their winning probability more frequently, typically for every block. However, these cryptocurrencies adopt more advanced algorithms to stabilize the block arrival rate at the targeted level. Therefore, in both Bitcoin and newer cryptocurrencies, the block arrival rate can be considered approximately constant over short periods. See \citet{noda2020difficulty,kawaguchi2022miners} for details about the algorithms for updating the lottery winning probability.} Drawing a single lottery corresponds to computing a function once, with the probability of winning each draw set to be extremely low. Miners use specialized hardware to perform this computation at extremely high speeds using specialized hardware while incurring operational costs such as electricity expenses. Accordingly, given that an individual miner keeps operating (irrespective of the block arrivals to other miners), the distribution of the time required for the miner to successfully add a block is well-approximated by an exponential distribution with an intensity proportional to their computational power. A new block is added to the blockchain when any one of the many miners all over the world wins the lottery, and the time required for it also follows an exponential distribution. Consequently, the block arrival process can be thought of as following a Poisson process.

Miners receive two types of rewards upon generating a block. The first is the transaction fees paid by users. Each user bids their own fee when submitting a transaction request, and miners collect these fees by including the transactions in their blocks. Thus, miners will prioritize transaction requests with higher fees when constructing blocks. The second is the block reward, a fixed amount paid to miners in newly minted coins upon block generation. As of 2025, the block reward is significantly larger than the total transaction fees collected. However, the block reward is halved every four years (a process known as ``halving'') and is programmed to eventually reach zero. Consequently, in the long run, transaction fees will become the sole source of revenue for miners. Accordingly, while the current dominance of block rewards means that transaction fees have little influence on miners' operational decisions and the miners typically operate at a constant level, in the future, miners may choose to suspend operations (or divert their machines to mine other cryptocurrencies) if the total transaction fees in the mempool are insufficient to justify mining costs.\footnote{It has been reported that miners adjust their level of effort in response to the rewards obtained from mining \citep{noda2020difficulty,kawaguchi2022miners}. This behavior may become even more pronounced in the future as block rewards decrease, reducing the minimum reward from mining.} Note that all rules of Bitcoin, including those related to halving, can be modified if there is consensus among stakeholders, particularly the miners who maintain the system, as long as the changes remain technically feasible.

\subsection{Related Literature}

Bitcoin has experienced rapid market growth, reaching a market capitalization of two trillion USD in 2025. Reflecting its increasing economic significance, extensive research has been conducted on various aspects of Bitcoin and other cryptocurrencies. One major strand of this literature examines Bitcoin's role as a decentralized medium of exchange \citep{garratt2018bitcoin,schilling2019some,matsushima2020}. Another line of research evaluates its effectiveness as a payment system and its potential for future development \citep{chiu2022economics,huberman2021monopoly,budish2025trust}. Additionally, the mining industry has been analyzed from multiple perspectives, including competition and collusion \citep{cong2021decentralized,dimitri2017bitcoin,arnosti2022bitcoin,capponi2023proof,lehar2020miner}. For a broader survey of the economic and financial literature, see \citet{john2022bitcoin}.

Among various topics on cryptoeconomics, transaction fees have emerged as an important area of research because Bitcoin's long-term viability depends on transaction fees, as the system is designed to eventually rely solely on these fees to incentivize miners. Early work by \citet{houy2014economics} already analyzed the impact of congestion on fees, using a partial equilibrium model. \citet{easley2019mining} use a game-theoretic model to establish a relationship between congestion and fees, highlighting that high fees and long waiting times may eventually discourage user participation. \citet{huberman2021monopoly} derive a closed-form relationship between fees and waiting times, comparing Bitcoin to traditional payment systems controlled by monopolists. \citet{ilk2021stability} empirically study users' demand and miners' supply in the Bitcoin market. \citet{hinzen2022bitcoin} demonstrate that an increase in potential users leads to higher congestion and fees, ultimately driving users away, suggesting an inherent limitation in Bitcoin's ability to sustain large-scale adoption. \citet{john2025proof} further argue that these issues could be fundamentally resolved by transitioning to an alternative consensus mechanism. A common limitation of these studies is that users are assumed to bid fees statically, without responding to short-term congestion. Our research offers a novel perspective on Bitcoin's fee dynamics by departing from this approach and explicitly modeling short-term congestion and its potential consequences.

The impact of declining block rewards on miner behavior has also been explored. \citet{carlsten2016instability} first point out (without a formal model) the risk that low fee revenues could lead miners to temporarily suspend operations. They further note that decreasing block rewards may induce an attack called ``fee sniping,'' potentially undermining Bitcoin’s security. Our study extends \citet{carlsten2016instability} by characterizing both users' equilibrium bids and miners' operational decisions, providing a formal analysis of these concerns and offering guidance on socially optimal block reward policies.

While we focus on analyzing Bitcoin’s existing fee mechanism, alternative designs have been widely discussed. 
\citet{Chen2019axiomatic} show that Bitcoin's method of issuing block rewards is the only one that satisfies certain desirable axioms.
\citet{basu2023towards} propose an auction inspired by uniform price auctions, whereas \citet{lavi2022redesigning} establish a new auction format called the monopolistic auction. Ethereum, the second-largest cryptocurrency, has been more proactive in revising its fee mechanism. In 2021, Ethereum adopted a proposal by \citet{eip1559} to introduce a base fee, dynamically adjusted based on prior block congestion, which is burned rather than collected by validators. \citet{roughgarden2021transaction,roughgarden2024transaction} analyzes this new mechanism, demonstrating its benefits in reducing fee volatility and simplifying optimal bidding strategies under normal demand conditions. \citet{ferreira2021dynamic} propose a dynamic posted-price mechanism, showing that it reduces price volatility further. The optimal design of transaction fee mechanisms is a crucial economic question, but Bitcoin’s developer community might be resistant to drastic protocol changes.\footnote{Since its launch in 2009, Bitcoin has never implemented a drastic rule change called hardfork. Even for upgrades that do not require a hardfork, Bitcoin requires approval from 95\% of miners. This stands in stark contrast to Ethereum, which regularly implements upgrades through hardforks, having conducted 19 such upgrades between 2016 and 2024.}
Minimal adjustments, such as sustaining the block reward as we propose, offer a practical and effective approach to addressing this real-world problem.\footnote{\citet{sonmez2023minimalist} argues that, when a central authority does not commission economists, policy recommendations that resolve issues with minimal intervention are effective. Since Bitcoin fundamentally operates under a decentralized governance structure, no authority exists.}

\section{User Competition}

This section considers a \emph{user-competition (UC) model}, a simple benchmark model in which (i) users are impatient and (ii) miners provide a constant hash rate irrespective of the current revenue from transaction fees. These assumptions are relaxed subsequently in Section \ref{sec:endo_ope}, where we consider an \emph{endogenous operation (EO) model} in which (ii) is replaced with (ii') miners need to pay a cost to operate a machine to produce a new block (so the hash rate may not be constant over time). Section~\ref{sec: patient user summary} and Appendix~\ref{sec: appendix patient users} further replace (i) with  (i') users are not perfectly impatient.

In this section, by imposing (i) and (ii), we focus on analyzing the competition between users through their bidding behavior. The insights from this section will be useful in considering a more complex environment studied in  Section \ref{sec:endo_ope}.

\subsection{Model}\label{subsec: model one-shot validation model with always-working miners}

Time runs continuously in $\mathbb{R}_+ = [0, +\infty)$. At time $0$, no transaction request has been made. Users arrive at the market continuously at rate $1$ (i.e., in any interval of size $a>0$, measure $a$ of new users arrive at the market).  Each user possesses one \emph{transaction request} and earns payoff $1$ when her transaction request is validated. Upon arrival, each user chooses her \emph{bid} (transaction fee) $b \in \mathbb{R}_+$. As we describe later, just like in a pay-as-bid auction, the user needs to pay her bid when her transaction request is validated. The bid cannot be modified subsequently. A user's request may not be (always) immediately validated, and therefore, the set of pending transactions (i.e., the set of users who have already arrived but still waiting for validation) forms a pool.

In this section, we assume that the block arrival process follows a homogeneous Poisson process with intensity $\lambda > 0$. Whenever there is a block arrival and the current measure of pending transactions is $L$, a mass $\min\{K, L\}$ of transaction requests are validated, where $K\in \mathbb R_{++}$ is a constant. The parameter $\lambda$ measures the frequency of block arrivals (this is on average every 10 minutes in the Bitcoin market).\footnote{In the actual implementation, Bitcoin dynamically adjusts a parameter called the difficulty to adjust the block arrival rate to the targeted level. In our model, the block arrival rate (given that miners operate) is assumed to be constant because (i) difficulty adjustment is not our focus, and (ii) Bitcoin adjusts difficulty only once per two weeks, and therefore, it is indeed constant in the short run (which is our focus).} The parameter $K$ corresponds to the block \emph{capacity}. Accordingly, $\lambda K$ represents the \emph{throughput} of the Bitcoin system. We assume that the validation of transaction requests follows a greedy process, i.e. when a block arrives, a transaction request with bid $b$ is validated if and only if, at that moment, the mass of transaction requests that have higher bids than $b$ is smaller than $K$.\footnote{This implies that, if multiple users make the same bid, all the marginal transaction requests are excluded from the block and the block size may become smaller than $K$. This assumption is made for analytical tractability and can be replaced with an alternative assumption, such as miners adopting uniform tie-breaking, without changing any model implication.}

Users are \emph{impatient}: a user's payoff is $1-b$ if the user's bid is included in the first block and it is $0$ otherwise. That is, the users are risk neutral and only care about the payoffs originating in the first block arrived. A \emph{patient-user model}, where users also care about the second block arrival onward, is discussed in Section~\ref{sec: patient user summary} and Appendix~\ref{sec: appendix patient users}.

A \textit{bid function} $\beta:\mathbb{R}_+\to\mathbb{R}_+$ is a measurable function that assigns a bid to each time. Call the user who arrives when the time is $t$ a \emph{time-$t$ user}. If all other users bid according to the bid function $\beta$ (i.e., bids $\beta(t')$ at time any $t'$), then the mass of transaction requests whose bid is higher than or equal to $b$ at time $t$ is given by
\begin{equation}\label{eq: model 1, L}
    L(t, b, \beta) = |\{s \in \mathbb{R}_+: s \le t \text{ and } \beta(s) \ge b\}|.
\end{equation}
A transaction request with bid $b$ is validated if and only if the block arrives before $L(t, b, \beta)$ becomes larger than or equal to $K$, or equivalently, a block arrives before time $\bar{t}(b; \beta)$, which is specified by
\begin{equation}\label{eq: model 1, bar t}
    \bar{t}(b; \beta) = \inf\{t \in \mathbb{R}_+: L(t, b, \beta) \ge K\}.
\end{equation}

Let $\pi(t, b, \beta)$ be the probability that a user's transaction request is validated given that the current time is $t$, the user's bid is $b$ and all the other users follow the bid function $\beta$. When $t < \bar{t}(b, \beta)$, then this is equal to the probability that a block does not arrive by time $\bar{t}(b, \beta)$, which is $1 - e^{- \lambda(\bar{t}(b;\beta) - t)}$. If $t \ge \bar{t}(b, \beta)$, the probability is zero. Accordingly, $\pi(t, b, \beta)$ can be summarized as follows.
\begin{equation}\label{eq: model 1, pi}
    \pi(t, b, \beta) = 1 - e^{- \lambda[\bar{t}(b; \beta) - t]^+}.
\end{equation}
Here, $[x]^+$ is equal to $x$ if $x \ge 0$ and $0$ otherwise.

\subsection{Equilibrium}

\subsubsection{Definition}

We say that a bid function $\beta$ is an equilibrium if, for all $t$, time-$t$ user has an incentive to bid $\beta(t)$, given that all the other users follow $\beta$. The definition is formally stated as follows.

\begin{defn}[Equilibrium, UC Model]\label{defn: model 1, equilibrium}
In the UC model, a bid function $\beta$ is an \emph{equilibrium} if for all $t \in \mathbb{R}_+$ and $b \in \mathbb{R}_+$,
\begin{equation}
    \pi(t, \beta(t), \beta)(1 - \beta(t)) \ge \pi(t, b, \beta)(1 - b).
\end{equation}
\end{defn}

The left-hand side is the expected payoff of the time-$t$ user who chooses her bid according to the bid function $\beta$, while the right-hand side is her expected payoff when she deviates and bids $b$. The definition of equilibrium deserves further explanations, and we detail them in the following remark.

\begin{remark}\label{remark: bid function}
\quad

\begin{enumerate}
\item\label{item:beta_infinitesmal}
In the optimality condition in the definition of equilibrium, even after a time-$t$ user deviates from the bid function $\beta$, every time-$t'$ user who arrives after her is assumed to bid $\beta(t')$. In other words, when a time-$t$ user calculates her expected payoffs from her possible bids, she believes that the subsequent play will be the same no matter which bids she makes. Given that the time-$t$ user is infinitesimal, ignoring the bid by the time-$t$ user is optimal for these subsequent users given that everyone else does so.
\item
The equilibrium object is a bid function, instead of a strategy profile. Even though $\beta$ specifies an action (a bid) for each time, it is easy to see that such a bid would not be optimal if the history of bids is not given by $\beta$.\footnote{For example, imagine a history at strictly positive $t>0$ in which all the past bidders have bid 0. The optimal bid would then be lower than the case in which the past users have followed $\beta$ that is strictly positive for almost all times before $t$.} In this sense, $\beta$ is not a full contingent plan (while it is sufficient to calculate the payoff not only on the path of play but also from a unilateral deviation as we explained in item (\ref{item:beta_infinitesmal}) of the current remark). We do not define strategies that specify a bid after every history because ensuring the existence of the best response at every history would be cumbersome.
\item
Besides the cumbersomeness of checking off-path best responses, there are technical reasons why we do not pursue dealing with fully contingent strategies. Our environment is one of continuous-time perfect-information games, which is known to entail technical issues in defining the strategy space. For example, a given strategy profile may induce zero or multiple action paths.\footnote{See \citet{simon1989extensive} and \citet{bergin1993continuous} point out issues with defining strategies in continuous time and offer some remedies in deterministic environments. 
\citet{KamadaRaoContinuous} consider a stochastic environment.} The complication is so severe that, for instance, requiring that a bid depends only on the current distribution of bids may not suffice to pin down the implication of a deviation on the continuation path of bids.\footnote{As an example, 
consider a strategy profile $\sigma$ in which, for each $t$, the bid by the time-$t$ user is $t$ as long as almost all time-$t'$ users with $t'<t$ have bid $t'$, while her bid is 0 if positive measure of type-$t'$ bidders with $t'<t$ have put different bids. This is a strategy that depends only on the distribution of bids. Suppose that all type-$t'$ bidders have bid $t'$ and the time-$1$ user deviated to bid 2. One possibility for the future action path under $\sigma$ is that all time-$t'$ users with $t'>1$ bids $t'$. Another possibility is that all subsequent users bid $0$. The latter is consistent with $\sigma$ because, in particular, for any $t'>1$, there is a positive measure of users who have bid differently from $\sigma$ (the users in $[1,t')$), and thus bidding 0 is consistent with $\sigma$.}
\end{enumerate}
\end{remark}

\subsubsection{Characterization}

In this section, we derive a closed form of the unique equilibrium bid function. As a first step, we first show some basic properties of the equilibrium bid function, which are useful for subsequent analyses.

\begin{thm}\label{thm: basic property of eqm one-shot always operating}
    In the UC model, if a bid function $\beta$ is an equilibrium, then the following conditions must be the case.
    \begin{enumerate}[(i)]
        \item $\beta(t) \in [0, 1)$ for all $t \in \mathbb{R}$.
        \item $\beta$ is strictly increasing.
        \item $\beta$ is continuous.
        \item $\beta(0) = 0$.
        \item $\lim_{t \to \infty} \beta(t) = 1$.
    \end{enumerate}
\end{thm}

The proofs are presented in Appendix~\ref{sec: proofs}. The intuition of Theorem~\ref{thm: basic property of eqm one-shot always operating} is as follows. The users' equilibrium bids reflect the severity of the competition. As $t$ increases and the measure of users continuously increases, the bid continuously increases and thus the competition becomes more and more severe. The bid starts at $0$ at time $0$ because the time-$0$ user's bid is the weakest in equilibrium anyway, and the bid approaches $1$ as the competition becomes infinitely harsh and users' surplus approaches zero in such a situation. The proofs for these properties are standard.

Theorem~\ref{thm: basic property of eqm one-shot always operating} implies that for all $b \in [0, 1)$, there exists $t$ such that $\beta(t) = b$; thus, to verify that $\beta$ is an equilibrium, it suffices to show that no time-$t$ user has an incentive to ``report'' $\hat{t} \neq t$, i.e., to bid $\beta(\hat{t})$. Furthermore, since $\beta$ is strictly increasing, transactions are included in a block in descending order up to capacity $K$. Hence, $\bar{t}(\beta(t), \beta) = t + K$ holds for all $t$.

Accordingly, time-$t$ user's payoff from reporting $\hat{t}$ is given by
\begin{equation}
    (1 - e^{-\lambda [\hat{t} - t + K]^+})(1 - \beta(\hat{t})).
\end{equation}
Note that the former term only depends on $s \coloneqq t - \hat{t}$. When $s\geq 0$, the value $s$ indicates the mass of users who have bid higher than $\beta(\hat{t})$ at time $t$. For $s < 0$, $-s$ represents the time until someone else bids $\beta(\hat{t})$, assuming the block will not arrive by then.

For notational convenience, we introduce a function $W$  as follows.
\begin{equation}\label{eq: W definition one-shot, always operating}
    W(s; \lambda, K) = \pi(t, \beta(t-s), \beta) = 1 - e^{-\lambda [K - s]^+}.
\end{equation}
We primarily treat $W$ as a function of $s$, but its shape also depends on parameters such as $\lambda$ and $K$. When there is no ambiguity, we omit the parameters listed after the semicolon and write $W(s)$ to simplify the notation. The same notational simplification will be applied to other functions as well.
Time-$t$ user's payoff from reporting $\hat{t} = t - s$ can be rewritten as
\begin{equation}
    W(s)(1 - \beta(t - s)).
\end{equation}
Taking the first-order condition, we can derive the user's optimality condition as a differential equation. 
There is a unique bid function that satisfies such an optimality condition, and we verify that the unique bid function is indeed an equilibrium.

\begin{thm}\label{thm: equilibrium closed form one-shot always operating}
In the UC model, there is a unique equilibrium. In this equilibrium, the bid function is specified by
\begin{equation}
    \beta(t) = \beta^E(t; \lambda, K) \coloneqq 1 - e^{\frac{W'(0; \lambda, K)}{W(0; \lambda, K)}t} = 1 - e^{-\frac{\lambda e^{- \lambda K}}{1 - e^{- \lambda K}}t}\label{eq: eqm bid function one-shot always operating}.
\end{equation}
\end{thm}

The shape of the equilibrium bid function $\beta$ is depicted in Figure~\ref{fig: user competition variation lambda and K}.

\subsubsection{Comparative Statics}

When the parameters $\lambda$ and $K$ are large, the system can process more transaction requests per unit time on average. This mitigates users' competition, resulting in lower equilibrium bids.

\begin{prop}
    In the UC model, for all $t \ge 0$, time-$t$ user's equilibrium bid $\beta^E(t; \lambda, K)$ is decreasing in $\lambda$ and $K$.
\end{prop}

The equilibrium bid function is presented in Theorem~\ref{thm: equilibrium closed form one-shot always operating}; thus, this theorem can be proven through direct calculation.

\begin{figure}[tb]
    \centering
    \begin{minipage}[t]{0.475\textwidth}
        \centering
        \includegraphics[width=\textwidth]{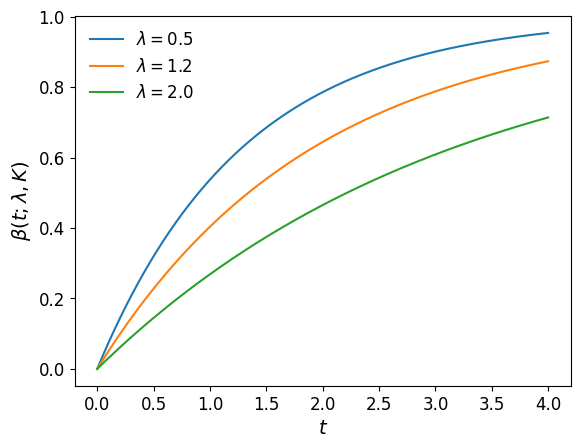}
        \subcaption{Bid functions for various block arrival rate $\lambda$}
        \label{fig: user competition variation lambda}
    \end{minipage}
    \hfill
    \begin{minipage}[t]{0.475\textwidth}  
        \centering 
        \includegraphics[width=\textwidth]{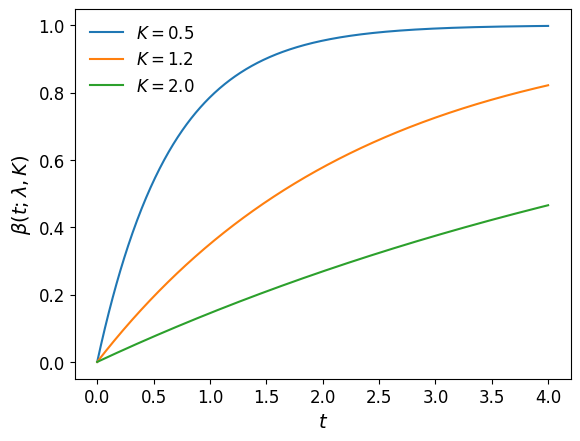}
        \subcaption{Bid functions for various block capacity $K$}
        \label{fig: user competition variation K}
    \end{minipage}
    \caption{Bid $\beta(t; \lambda, K)$ as a function of $t$.}
    \begin{center}\footnotesize
        Note: For Figure~\ref{fig: user competition variation lambda}, we fix $K = 1.0$. For Figure~\ref{fig: user competition variation K}, we fix $\lambda = 1.0$.
    \end{center}
    \label{fig: user competition variation lambda and K}
\end{figure}

\subsection{Welfare}\label{subsec: user competition welfare}

\subsubsection{User's Welfare}

We evaluate users' welfare by ex ante expected payoff. To define users' welfare formally, we first define the time-$t$ user's equilibrium expected payoff, $U(t)$, as follows.
\begin{align}
    U(t) &= \int_0^K (1 - \beta^E(t)) \lambda e^{-\lambda t'}dt' \label{eq: defn U(t)}\\
     &= \left(1 - e^{-\lambda K} \right) e^{- \frac{\lambda e^{-\lambda K}}{1 - e^{-\lambda K}}t}.
\end{align}
In equilibrium, a time-$t$ user bids $\beta^E(t)$; thus she enjoys a payoff of $1 - \beta^E(t)$ when her transaction is validated. Since the equilibrium bid function $\beta^E$ is strictly increasing, this user's transaction is validated if and only if a block arrives in a time interval $[t, t + K)$. From a time-$t$ user's perspective, this event occurs with probability $\int_0^K \lambda e^{-\lambda t'}dt' = 1 - e^{-\lambda K}$. Accordingly, the time-$t$ user's equilibrium expected payoff, $U(t)$, is specified by \eqref{eq: defn U(t)}.

The following properties are immediate from the closed-form representation of $U$.
\begin{prop}
In the UC model, fr any $t\geq 0$, time-$t$ user's equilibrium expected payoff, $U(t)$, satisfies the following properties:
    \begin{enumerate}[(i)]
        \item $U(t)$ is decreasing in $t$.
        \item $U(t; \lambda, K)$ is increasing in $\lambda$ and $K$.
    \end{enumerate}
\end{prop}

The time-$t$ user's expected payoff $U(t)$ is decreasing in $t$ because the equilibrium bid function $\beta$ is increasing in $t$. As more users arrive, users face severer competition and need to pay more fees to secure the same probability of validation. Consequently, as $t$ increases, the user's surplus shrinks.

We consider a Markov chain in which time $t$ continuously runs forward and then resets to zero upon each block arrival. Using the stationary distribution of this chain, we evaluate the surplus of users and miners. Specifically, we define users' welfare $\bar{U}$ as an expectation of $U(t)$, where time $t$ follows its stationary distribution. The stationary distribution of $t$ is characterized as follows.

\begin{prop}\label{prop: stationary distribution}
    In the UC model, there is a unique stationary distribution of $t$. Furthermore, it is an exponential distribution with intensity $\lambda$, whose cumulative distribution function is given by $\Psi(t) = 1 - e^{- \lambda t}$.
\end{prop}

The user's welfare $\bar{U}$ is defined as
\begin{equation}
    \bar{U} = \mathbb{E}_{t \sim \Psi}[U(t)]
    =\int_0^\infty U(t) \lambda e^{-\lambda t} dt
    = \lambda \left(1 - e^{-\lambda K} \right) \int_0^\infty e^{- \frac{\lambda}{1 - e^{-\lambda K}} t}dt
    = \left(1 - e^{-\lambda K} \right)^2.\label{eq: user competition bar U}
\end{equation}

The following property is immediate from \eqref{eq: user competition bar U}.

\begin{prop}
    In the UC model, the throughput $\lambda K$ is a sufficient statistic to determine the users' welfare $\bar{U}$. Furthermore, $\bar{U}$ is increasing in $\lambda K$.
\end{prop}

When $\lambda$ and $K$ increase, blocks arrive more frequently, and users' transactions are more likely to be validated. Furthermore, frequent block arrivals mitigate user competition, reducing users' transaction fees users will pay. Accordingly, users' welfare is increasing in $\lambda$ and $K$.

\subsubsection{Miner's Revenue}

Next, we evaluate the miner's revenue. We define $R(t)$ as the expected payment a time-$t$ user will make, where the expectation is taken at the timing of the user's arrival.
\begin{align}
    R(t) &= \int_0^K \beta(t) \lambda e^{-\lambda t'}dt' \label{eq: user competition time-t payment}\\
    &= \left(1 - e^{-\lambda K} \right) \left(1 - e^{- \frac{\lambda e^{-\lambda K}}{1 - e^{-\lambda K}}t}\right).
\end{align}
When validated, a time-$t$ user pays $\beta(t)$ to the miner. This happens if a block arrives in a time interval $[t, t + K)$, which occurs with probability $1 - e^{-\lambda K}$. Accordingly, a time-$t$ user's expected payment $R(t)$ is given by \eqref{eq: user competition time-t payment}.

The miner's revenue, denoted $\bar{R}$, is defined as an average flow revenue collected from users, i.e.,
\begin{equation}
    \bar{R} = \mathbb{E}_{t \sim \Psi}[R(t)] = e^{-\lambda K} (1 - e^{-\lambda K}).
\end{equation}

\begin{prop}
    In the UC model, the throughput $\mu = \lambda K$ is a sufficient statistic to determine the miner's revenue $\bar{R}$. $\bar{R}$ is increasing in $\mu$ for $\mu < \log 2$ and decreasing in $\mu$ for $\mu > \log 2$.
\end{prop}

Unlike the user's welfare $\bar{U}$, the miner's revenue $\bar{R}$ is not monotonic in the throughput $\mu = \lambda K$ because an increase in $\mu $ impacts the miner's revenue through two channels. First, it allows the miner to validate more transactions more frequently, increasing social welfare. Second, it mitigates user competition and decreases the equilibrium bid $\beta^E$. The former effect increases the miner's revenue, whereas the latter effect decreases it. When $\mu$ is small, the former effect is dominant, while the latter is dominant when $\mu$ is large. Consequently, the miner's revenue is not monotonic.

\section{Endogenous Operation}\label{sec:endo_ope}

\subsection{Model}

Thus far, we have assumed that block arrivals follow the Poisson process with a fixed arrival rate of $\lambda$. This assumption may be invalid in a more realistic situation where miners strategically choose whether to operate their machines to produce blocks. To see this, note that in reality, miners need to pay a cost to operate a machine to produce a new block. Moreover, when the machine can be diverted for other purposes (e.g., mining different cryptocurrencies), mining also incurs an opportunity cost. If miners find it unprofitable to mine a cryptocurrency, then the cryptocurrency's block arrival rate will decline. Hence, if the profitability of mining changes over time, the arrival rate of the blocks would also vary over time.

In this section, we consider an \emph{endogenous operation (EO) model}, a model with a continuum of infinitesimal miners who pursue their own profit. We consider two types of miners: \emph{committed miners} and \emph{switching miners}. We denote the mass of committed miners by $\eta \in [0, 1]$. Each committed miner always operates her machine. The remaining mass, $1 -\eta$, consists of switching miners, who decide moment by moment whether to operate their machines, incurring no switching costs.\footnote{We do not define the switching miners' strategies as a mapping from histories to actions to avoid technical complications in defining strategies in continuous time. Instead, we consider a mapping from times to actions, as we specify shortly.} 

Our interest is in the case of $\eta = 0$ where all miners are strategic, but we consider the case with $\eta>0$. The reason is that, as we will discuss in detail later, there exist infinitely many uninteresting equilibria when $\eta=0$, and considering the cases with $\eta>0$ helps us ``select'' a unique one. The selection is made by concentrating on the equilibrium that we obtain as the ``limit'' of the equilibria in the $\eta > 0$ case, which we show to be unique for each $\eta>0$. To formalize this argument, we allow for $\eta > 0$ in the subsequent discussion.

When mass $\Delta$ of miners operate, they produce a new block with a Poisson arrival rate of $\lambda \Delta$ in total, and they pay a flow cost of $c \Delta$ (where $c \ge 0$) in total. When they suspend, they have no chance to produce a new block and pay no cost.

Since committed miners have no action choice, we study switching miners' profit maximization problem. Since all switching miners have an identical cost structure, their aggregate behavior can be derived by solving a representative switching miner's decision problem. Hereafter, we equate all switching miners with the representative switching miner, and express them as ``the miner.'' The miner has a mass of $1 - \eta$. When she operates, she finds a new block with a Poisson arrival rate of $\lambda (1 - \eta)$, while she needs to pay a flow cost of $c(1-\eta) \ge 0$. When the miner suspends, she never finds a block, while she pays no flow cost.

The miner's \emph{operation function} $\sigma: \mathbb{R}_+ \to \{0, 1\}$ is a measurable function that assigns the miner's binary action to each time. Parallel to the bid function $\beta$, we define the operation function as a function from times to operation decisions. Our intention for this specification is analogous to that explained in Remark~\ref{remark: bid function}. The decision $\sigma(t) = 1$ means the miner operates at time $t$, and $\sigma(t) = 0$ means she suspends. Since mass $\eta$ of committed miners always operate, the block arrival follows a nonhomogeneous Poisson process with a time-dependent intensity $\lambda(t) = \lambda (\eta + (1 - \eta) \sigma(t))$.

From the users' perspective, the only difference between the UC model and the EO model is the block arrival process. Accordingly, $L(t, b, \beta)$ and $\bar{t}(b; \beta)$ are specified by \eqref{eq: model 1, L} and \eqref{eq: model 1, bar t}, whereas the miner's EO induces a different underlying bid function, $\beta$. Since block arrivals follow a non-homogeneous Poisson process, the probability that a user's bid is validated, $\pi(t, b, \beta)$, should be modified as follows:
\begin{equation}\label{eq: model 2, pi}
    \pi(t, b; \beta, \sigma) = 1 - e^{- [A(\bar{t}(b, \beta); \sigma) - A(t; \sigma)]^+},
\end{equation}
where
\begin{equation}
    A(t; \sigma) \coloneqq \lambda\int_0^t (\eta + (1 - \eta) \sigma(s)) ds.
\end{equation}
Since we assume $\eta > 0$, $A(t;\sigma)$ is strictly increasing in $t$ for any $\sigma$.
Given this $\pi$, time-$t$ user's expected profit can be expressed as $\pi(t, b; \beta, \sigma) (1 - b)$ if she chooses bid $b$.

When the miner produces a new block at time $t$, she validates mass $\min\{K, t\}$ transaction requests greedily. Accordingly, by producing a block at time $t$, the miner collects
\begin{equation}
    M(t; \beta) = \int_{s \in [0, t]: L(s, \beta(s), \beta) \ge K} \beta(s) ds
\end{equation}
as the sum of transaction fees. In addition, the miner receives a block reward of $y \ge 0$. Thus, upon producing a block at time $t$, the miner obtains $M(t; \beta) + y$ in total. 
The miner's flow expected payoff is $(1 - \eta) (\lambda (M(t; \beta) + y) - c)$ if she operates, and $0$ if she suspends.

\subsection{Equilibrium}

\subsubsection{Definition}

We define equilibrium as a pair of a bid function $\beta$ and an operation function $\sigma$ that 
they jointly satisfy the user's and miner's optimality conditions.

\begin{defn}[Equilibrium, EO Model]\label{defn: model 2, equilibrium}
In the EO model, a pair of a bid function $\beta$ and an operation function $\sigma$ is an \emph{equilibrium} if the following two conditions are satisfied:
\begin{description}
    \item[User's Optimality] For all $t \in \mathbb{R}_+$ and $b \in \mathbb{R}_+$,
    \begin{equation}
        \pi(t, \beta(t); \beta, \sigma)(1 - \beta(t)) \ge \pi(t, b; \beta, \sigma)(1 - b).
    \end{equation}
    \item[Miner's Optimality] For all $t \in \mathbb{R}_+$,
    \begin{equation}
        \sigma(t) \in \argmax_{x \in \{0, 1\}} x (\lambda (M(t; \beta) + y) - c).
    \end{equation}
\end{description}
\end{defn}

Assuming no discounting, the miner's expected payoff from operating for a small interval $\Delta t$ starting at time $t$ is given by $\lambda \Delta t (M(t; \beta) + y) - c \Delta t + o(\Delta t)$. 
If $\lambda (M(t; \beta) + y)  c \neq 0$, then there exists $\bar{\Delta} > 0$ such that for all $\Delta t \in (0, \bar{\Delta})$, the sign of this miner's expected payoff coincides with the sign of $\lambda (M(t; \beta) + y) - c$, which quantifies the instantaneous profit from operation. Accordingly, in our equilibrium definition, a miner operates at time $t$ if doing so is immediately profitable, and suspends otherwise. This is what we define as the miner's optimality condition.

The underlying assumptions behind this definition are twofold: first, each individual miner is infinitesimal and acts as a price taker; second, there are no switching costs associated with transitioning between operation and suspension. Even if the miner is not myopic by nature and takes future payoffs into account, because each individual miner is infinitesimal, their operational decisions do not affect their payoffs from future blocks, which is driven by the aggregate dynamics of the block arrival process. Moreover, since switching costs are assumed to be absent, each miner can adjust their operational status at any moment. Consequently, (possible) concerns about future payoffs do not affect the miner's decision-making at each point in time and the miner can be assumed to myopically maximize her flow payoff.

\subsubsection{Characterization}

This section characterizes the equilibria under the EO model. First, we show that the basic properties stated in Theorem~\ref{thm: basic property of eqm one-shot always operating} are maintained.

\begin{thm}\label{thm: basic property of eqm Model 2}
    Suppose $\eta>0$ and fix an operation function $\sigma$ arbitrarily. If a bid function $\beta$ satisfies the user's optimality condition of the EO model given $\sigma$, then the following hold.
    \begin{enumerate}[(i)]
        \item $\beta(t) \in [0, 1)$ for all $t \in \mathbb{R}_+$.
        \item $\beta$ is strictly increasing.
        \item $\beta$ is continuous.
        \item $\beta(0) = 0$.
        \item $\lim_{t \to \infty} \beta(t) = 1$.
    \end{enumerate}
\end{thm}

Parallel to Theorem~\ref{thm: basic property of eqm one-shot always operating}, Theorem~\ref{thm: basic property of eqm Model 2} implies that for all $b \in [0, 1)$, there exists a unique $t \in \mathbb{R}_+$ such that $\beta(t) = b$. Since for all $t$, a time-$t$ user obtains a strictly positive expected payoff by bidding $\beta(t) \in [0, 1)$, we do not have to consider a deviation to $b \ge 1$. Thus, a bid function $\beta$ is an equilibrium if and only if for all $t$, the time-$t$ user has no incentive to ``report'' $\hat{t} \neq t$, i.e., to bid $\beta(\hat{t})$.

Part (ii) of Theorem~\ref{thm: basic property of eqm Model 2} implies that, when $\beta$ satisfies the user's optimality condition, the sum of transaction fees the miner can collect, $M(\cdot; \beta)$, is strictly increasing. In equilibrium, the miner optimally responds to such a bid function, and therefore, the miner takes a \emph{threshold strategy}, with which the miner suspends before a certain threshold time $t^*$ and operates afterward.

\begin{lem}\label{lem: optimality of threshold strategy}
    In the EO model, if $(\beta, \sigma)$ is an equilibrium, then $\sigma$ is a \emph{threshold strategy}, i.e., satisfies the following property:
    There exists a \emph{threshold time} $t^* \in \{-\infty, +\infty\} \cup \mathbb{R}_+$ such that $\sigma(t) = 0$ if $t < t^*$ and $\sigma(t) = 1$ if $t > t^*$.
\end{lem}

Given that the miner adopts a threshold strategy, its threshold time almost fully specifies the miner's behavior and it fully specifies the block arrival process. Slightly abusing the notation, hereafter we represent an operation function $\sigma$ by its threshold time $t^*$.

In the following, we characterize the equilibria under the EO model.

\begin{thm}\label{thm: equilibrium closed form model 2}
In the EO model, there exists an equilibrium. Furthermore, $(\beta, t^*)$ is an equilibrium if it satisfies the following conditions:
\begin{enumerate}[(i)]
    \item $t^* = t^E$, where $t^E$ is equal to (a) $-\infty$ if $\lambda y > c$, (b) $+ \infty$ if $\lambda (K + y) \le c$, and (c) otherwise, $t \in \mathbb{R}_+$ satisfying the following equality:
    \begin{equation}\label{eq: miner's threshold model 2}
        \lambda(M(t; \beta) + y) = c.
    \end{equation}
    Furthermore, there is at most one $t$ satisfying \eqref{eq: miner's threshold model 2}.
    \item The bid function $\beta$ is specified by $\beta(t) = \beta^E(t, t^E)$, where $\beta^E: \mathbb{R}_+^2 \to \mathbb{R}_+$ is
    \begin{equation}\label{eq: eqm bid function model 2}
    \beta^E(t, t') \coloneqq 1 - e^{\int_0^t \frac{W_1(0, t' - \tau)}{W(0, t' - \tau)}d\tau},
    \end{equation}
    and
    \begin{equation}
    W(s, l) = \begin{cases}
        1 - e^{- \eta \lambda (K - s)} & \text{ if } K - s \le l\\
        1 - e^{- \eta \lambda l - \lambda (K - s - l)} & \text{ if } 0 \le l < K - s\\
        1 - e^{- \lambda (K - s)} & \text{ if } l < 0,
    \end{cases}
    \end{equation}
    where $W_1$ is the partial derivative of $W$ with respect to $s$.
\end{enumerate}
Such $(\beta^E, t^E)$ that satisfies conditions (i) and (ii) is unique.
Furthermore, for $\eta > 0$, the equilibrium described above is a unique equilibrium.
\end{thm}

Note that, whenever $\beta$ is strictly increasing, the probability that a time-$t$ user's transaction request is validated when she pretends to be a time-$(t - s)$ user is equal to $W(s, t^* - t)$, i.e.,
\begin{equation}
    \pi(t, \beta(t - s); \beta, t^*) = W(s, t^*-t).
\end{equation}
The value $W(s, l)$ can be interpreted as the probability that a user's transaction is validated if the current time is $t$, this user's bid is $\beta(t - s)$, and the time when the miner starts working is $t + l$. This probability depends neither on the detailed shape of the $\beta$ function nor the current time $t$.

As stated in Theorem~\ref{thm: equilibrium closed form model 2}, the equilibrium shown in the theorem is unique when $\eta > 0$.\footnote{That is, we could replace ``if'' in the statement of the theorem to ``if and only if'' in this case.} In contrast, when $\eta= 0$, there are infinitely many equilibria. One example arises when the block reward $y$ is sufficiently small, where every user bids zero (or any small enough amount) and the miner never operates. In this situation, no unilateral deviation by any agent can yield a profit. Such equilibria do not arise when $\eta > 0$, which guarantees the existence of an operating miner at all times. 
Another possibility is that the miner uses the same operation function while the bid function is different before the threshold time. Such a profile can be an equilibrium because only the bid distributions at the threshold time onward matter as the miner never works before that. For example, if $(\hat{\beta},t^E)$ such that $\hat{\beta}(t)=\beta^E(t^*-t)$ for all $t\leq t^*$ and $\hat{\beta}(t)=\beta^E(t)$ for all $t> t^*$ is an equilibrium. 

Given these kinds of multiplicity, for $\eta = 0$, we regard the bid function and threshold time represented as the $\eta \to 0$ limit of the unique equilibria as the equilibrium of interest under $\eta = 0$.

\subsubsection{Comparative Statics}

\paragraph{Bid Function}

First, without imposing any condition on $\eta$, we can show that the equilibrium bid function $\beta^E$ is increasing in the threshold time, $t^*$.
\begin{thm}\label{thm: EO model beta is increasing in threshold}
    In the EO model, for all $t,t',t''\in \mathbb R_+$ such that 
    $t'< t''$, we have $\beta^E(t, t') \le \beta^E(t, t'')$.
\end{thm}

\begin{figure}[tbp]
    \centering
    \includegraphics[width = 0.5\textwidth]{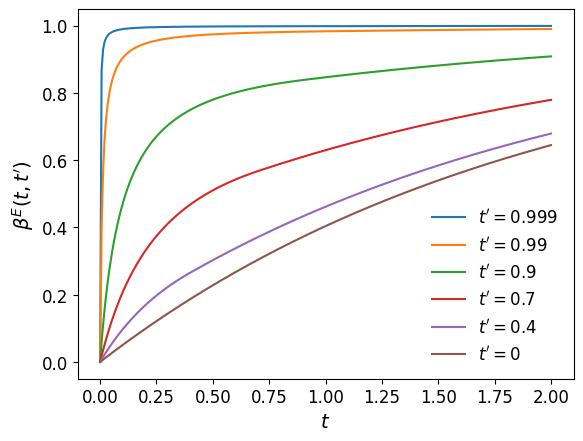}
    \caption{The shape of the equilibrium bid function $\beta^E(t, t')$ with varying threshold time $t'$.}
    \label{fig: beta EO threshold}
    \begin{center}\footnotesize
        Parameters: $\lambda = 1.2$, $K = 1$, $\eta = 0$. 
    \end{center}
\end{figure}

Users who arrive before the miner starts working have to compete against all the bidders in the same situation at the threshold time. The competition becomes more severe as the miner's absence becomes longer. Accordingly, the equilibrium bid function is increasing in the threshold time.

In fact, we have a slightly stronger condition---in the proof, we show that the hazard rate of $\beta^E(t,t')$ with respect to $t$ is increasing in $t'$ given any $t$. It is well known that the hazard-rate dominance condition implies the first-order stochastic dominance condition, which is stated as Theorem~\ref{thm: EO model beta is increasing in threshold}. See the proof of Theorem~\ref{thm: EO model beta is increasing in threshold} for the detail.

Next, we will show that when $\eta = 0$, $t^E \in [0, K)$ must be the case. This fact can be verified without inspecting the detailed shape of $(\beta^E, t^E)$ characterized by Theorem~\ref{thm: equilibrium closed form model 2}. Consider any $(\beta, t^*)$ such that $t^* \geq K$. When $\eta=0$, a block arrives only after the threshold time $t^*$. Since $(\beta, t^*)$ can be an equilibrium only if $\beta$ is increasing, users arriving from time $0$ to time $t^* - K$ have no chance to win; thus, they want to increase their bid. However, in an equilibrium, $\beta(0) = 0$ must also be the case. Thus, such $(\beta, t^*)$ cannot be an equilibrium. The following lemma formalizes this argument.

\begin{lem}\label{lem: tE in 0 K}
    In the EO model, suppose that $\lambda y \le c < \lambda (K + y)$. For $\eta = 0$, we have $t^E \in [0, K)$.
\end{lem}

Given $\eta = 0$, we can derive a closed-form expression of the equilibrium bid function $\beta^E$ for threshold times in $[0,K)$ as follows.

\begin{lem}\label{thm: model 2 beta function limit}
    In the EO model, for $\eta = 0$, for $t^* \in [0, K)$,
    \begin{equation}
        \beta^E(t, t^*) = \begin{cases}
            1 - \dfrac{1 - e^{-\lambda(K - t^*)}}{1 - e^{-\lambda(K + t - t^*)}} & \text{ if } t < t^*;\vspace{0.5em}\\
            1 - \dfrac{1 - e^{-\lambda(K - t^*)}}{1 - e^{-\lambda K}} e^{- \frac{\lambda e^{- \lambda K}}{1 - e^{- \lambda K}}(t -t^*)} & \text{ otherwise.}
        \end{cases}
    \end{equation}
\end{lem}

Next, we evaluate the value of $\beta^E(t, t^*)$ for $t \le t^*$ as follows.
\begin{lem}\label{lem: beta increasing in tstar}
    For $\eta = 0$, for $t \le t^*$, $\beta^E(t, t^*)$ is strictly convex in $t^*$.
\end{lem}

The intuition for the convexity is as follows.
Larger $t^*$ implies that the time-$t$ user faces severer competition at time $t^*$ when the miner starts operating, leading to a higher transaction fee. Furthermore,  more users attempt to increase transaction fees. As a result, the effect of severer competition becomes more pronounced when $t^*$ is larger.

\paragraph{Threshold Time}

We consider how the equilibrium threshold time, $t^E$, responds to the parameter changes. By Theorem~\ref{thm: equilibrium closed form model 2}, we already know that $t^E = - \infty$ if $\lambda y > c$ and $t^E = + \infty$ if $\lambda (K + y) \le c$. We will focus on the case of $\lambda y \le c < \lambda (K + y)$, with which the equilibrium threshold time $t^E$ satisfies $\lambda (M(t^E; \beta^E(\cdot; t^E)) + y) = c$.

The behavior of $t^E$ with respect to the parameters of the model is determined by the miner's surplus, $M^*(t^*) \coloneqq M(t^*; \beta^E(\cdot; t^*)) = \int_{[t^* - K]^+}^{t^*}\beta^E(t, t^*) dt$, as a function of $t^*$. The following proposition shows the basic property of the $M^*$ function.

\begin{prop}\label{thm: comparative statics M}
    In the EO model, the miner's surplus $M^*$ is increasing. Furthermore, for $\eta = 0$ and $t^* \in [0, K]$, $M^*$ is strictly convex.
\end{prop}

\begin{figure}[tb]
    \centering
    \begin{minipage}[t]{0.475\textwidth}
        \centering
        \includegraphics[width=\textwidth]{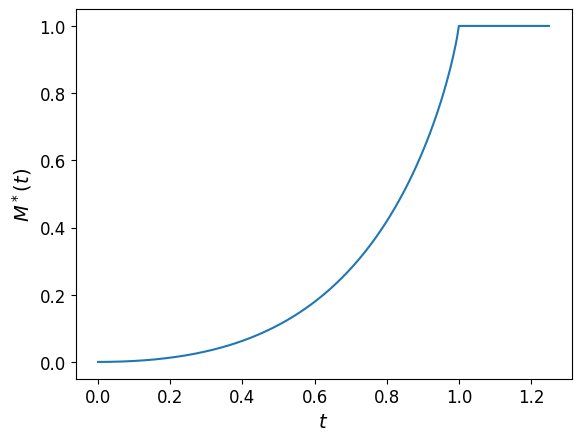}
        \subcaption{The miner's surplus at threshold time $t$, $M^*(t)$}
        \label{fig: Mstar}
    \end{minipage}
    \hfill
    \begin{minipage}[t]{0.475\textwidth}  
        \centering 
        \includegraphics[width=\textwidth]{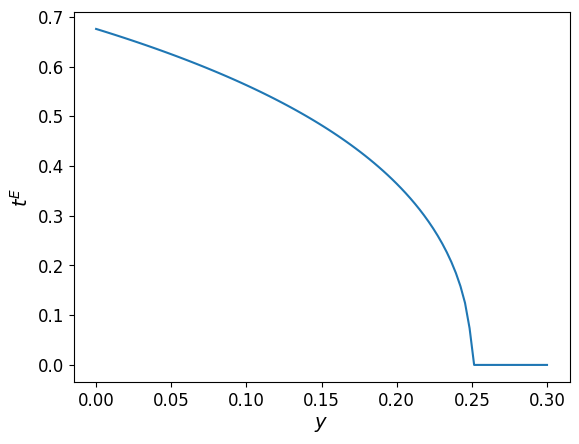}
        \subcaption{The equilibrium threshold time $t^E$}
        \label{fig: tE}
    \end{minipage}
    \caption{(a) The miner's surplus at threshold time $t$, $M^*(t)$, and (b) the equilibrium threshold time $t^E$ with varying block reward $y$}
    \label{fig: Mstar and tE}
    \begin{center}\footnotesize
        Parameters: $\lambda = 1.2$, $K = 1$, $c = 0.3$, $\eta = 0$. 
    \end{center}
\end{figure}

The shape of $M^*$ is depicted in Figure~\ref{fig: Mstar}.\footnote{For $\eta = 0$ and $t^* \in [0, K]$, we can compute the closed form of $M^*$ by direct calculation. See Appendix~\ref{subsec: M closed form limit} for the details.} The value $M^*(t^*)$ represents the miner's revenue when a block arrives at time $t^*$ given that the miner produces a new block at time $t^*$. As $t^*$ increases, (i) there are more bidders at time $t^*$ if the block capacity is not yet full, and (ii) the bidders make higher bids as users face severer competition (Theorem~\ref{thm: EO model beta is increasing in threshold}). Clearly, both effects increase $M^*(t^*)$. Furthermore, for $\eta = 0$ and $t^* \le K$, the users' bid is convex in the threshold time, implying that $M^*$ is also convex.

Recall that when the equilibrium threshold time satisfies $t^E \in [0, K)$, it solves $\lambda (M^*(t^E) + y) = c$, or equivalently,
\begin{equation}\label{eq: tE determination}
    M^*(t^E) = \frac{c}{\lambda} - y.
\end{equation}
Furthermore, the block reward $y$ and the cost of mining $c$ do not directly influence the shape of $\beta^E$. Therefore, we can apply the implicit function theorem to analyze the response of $t^E$ against the changes in $y$ and $c$.

\begin{prop}\label{prop: EO model tE response y and c}
    In the EO model, for $\eta = 0$ and $t^E \in [0, K)$, the equilibrium threshold time $t^E$ is (i) decreasing and strictly concave in $y$, and (ii) increasing and strictly convex in $c$.
\end{prop}

The shape of $t^E$ as a function of $y$ is depicted in Figure~\ref{fig: tE}. When the block reward $y$ increases, mining becomes more profitable. Thus, the miner starts operation earlier by decreasing the equilibrium threshold time $t^E$. While the miner's profit is linear in $y$, given $t \le K$, the miner's profit from transaction fees, $M^*(t)$, is convex, implying that the equilibrium threshold time $t^E$ should be concave in $y$. Since the changes in $y$ and $- c/\lambda$ provide the same effect, we have the same conclusion for the changes in $ - c$.

We will show that the equilibrium threshold time $t^E$ is decreasing in $K$.

\begin{lem}\label{lem: beta E is decreasing in K}
In the EO model, for $\eta = 0$, for $t \le t^*$, the equilibrium bid $\beta^E(t, t^*; K)$ is decreasing in $K$.
\end{lem}

The parameter $K$ represents the upper bound of the number of transactions validated with a block arrival. An increase in $K$ naturally mitigates users' competition, leading to the reduction of their bid.

\begin{lem}\label{thm: M is decreasing in K}
    In the EO model, for $\eta = 0$ and $t^* \in [0, K)$, the miner's surplus $M^*(t^*; K)$ is decreasing in $K$.
\end{lem}

As $t^E \in [0, K)$, it suffices to evaluate $M^*$ for $t^* \in [0, K)$. Given this, an increase in $K$ does not directly increase $M^*(t^*)$. Nevertheless, Lemma~\ref{lem: beta E is decreasing in K} implies that the miner will suffer from the reduction in transaction fees. Consequently, fixing $t^*$, the miner's revenue at the threshold time $t^*$ decreases as $K$ increases. Therefore, the equilibrium threshold time, $t^E$, is increasing in $K$.

\begin{prop}\label{prop: tE is increasing in K}
    In the EO model, for $\eta = 0$, the equilibrium threshold time $t^E$ is increasing in $K$.
\end{prop}

Finally, we analyze how the block arrival rate $\lambda$ influences the equilibrium outcome. Similarly to the impact from block capacity $K$, an increase in the block arrival rate $\lambda$ mitigates user competition and lowers the equilibrium bid function $\beta^E$.

\begin{lem}\label{lem: beta is decreasing in lambda}
    In the EO model, for $\eta = 0$, for $t \le t^*$ and $t^* \in [0, K)$, the equilibrium bid $\beta^E(t, t^*; \lambda)$ is decreasing in $\lambda$.
\end{lem}

This fact immediately implies that $M^*(t^*; \lambda)$ is decreasing in $\lambda$.

\begin{lem}
    In the EO model, for $\eta= 0$ and $t^* \in [0, K)$, the miner's surplus $M^*(t^*, \lambda)$ is decreasing in $\lambda$.
\end{lem}

However, unlike the comparative statics for block capacity $K$, the effect of the block arrival rate $\lambda$ on the equilibrium threshold time $t^E$ is not monotonic. 

\begin{thm}\label{thm: tE is nonmonotonic in lambda}
    In the EO model, for $\eta = 0$, $y = 0$ and $c > 0$, the equilibrium threshold time $t^E$ is not monotonic in $\lambda$.
\end{thm}

The intuition for the non-monotonicity is based on two effects of the increase of $\lambda$ on the left-hand side of the equation characterizing $t^E$, $\lambda (M^*+ y) = c$.
First, there is a direct effect, corresponding to the increase of the multiplier on $M^*+y$ in the left-hand side. 
Here, increasing $\lambda$ reduces the average cost of generating one block, $c/\lambda$, which decreases $t^E$. 
Second, there is an indirect effect such that $\lambda$ affects $M^*$. Here, increasing $\lambda$ lowers the equilibrium bid function $\beta^E$, which reduces the profit from block generation $M^*$, exerting upward pressure on $t^E$.
The relative magnitudes of the two effects vary depending on the value of $\lambda$. In particular, when $y = 0$, we can show that $t^E$ is decreasing for small $\lambda$ and increasing for large $\lambda$. 

\begin{figure}[tbp]
    \centering
    \includegraphics[width = 0.5\textwidth]{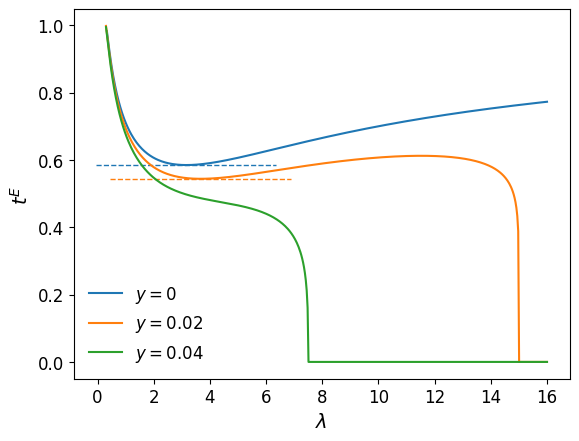}
    \caption{Block arrival rate $\lambda$ and threshold time $t^E$.}
    \label{fig: effect of lambda on tE}
    \begin{center}\footnotesize
        Parameters: $K = 1$, $c = 0.3$, $\eta = 0$. 
    \end{center}
\end{figure}

It follows from the continuity of $M^*$ and the equation characterizing $t^E$, \eqref{eq: tE determination}, that $t^E$ is continuous in the block reward $y$. Therefore, Theorem~\ref{thm: tE is nonmonotonic in lambda} implies that $t^E$ is non-monotonic in $\lambda$ when $y > 0$ is sufficiently small. However, it remains unclear whether such non-monotonicity holds for general $y > 0$. This is because, when $y > 0$, even if $M^*$ becomes zero, miners still receive the reward $y > 0$ for a block generated, thus the second indirect effect does not remain dominant at high values of $\lambda$. Consequently, when $y > 0$, $t^E = 0$ holds for sufficiently large $\lambda$, and the non-monotonicity of $t^E$ cannot be proven using the approach in Theorem~\ref{thm: tE is nonmonotonic in lambda}.

Figure~\ref{fig: effect of lambda on tE} illustrates how threshold time $t^E$ responds to changes in $\lambda$. As demonstrated in Theorem~\ref{thm: tE is nonmonotonic in lambda}, when block reward $y = 0$, $t^E$ is nonmonotonic: it is decreasing in $\lambda$ when $\lambda$ is small and is increasing when $\lambda$ is large. For $y > 0$, $t^E$ conversely approaches zero as $\lambda$ becomes large. As depicted in the case of $y = 0.02$, threshold time $t^E$ is still monotonic as there is an interval in which $t^E$ is increasing in $\lambda$. However, with large $y$ like $y = 0.04$, $t^E$ becomes monotonically decreasing for all $\lambda$.

\subsection{Welfare}

In this section, we analyze social welfare achieved under the EO model and explore how adjusting parameters such as the block reward $y$ can enhance welfare. In Section~\ref{subsec: user competition welfare}, we evaluated the user surplus and the miner revenue separately. However, it is unclear which agents in the economy bear the cost of newly minting coins (i.e., printing additional money of amount $y$), whereas the effect from the block reward $y$ is essential for the EO model. This section focuses on the evaluation of social welfare, which can be defined independently from the utility transfer through minting coins.

Parallel to the UC model, we consider a Markov chain in which time $t$ continuously runs forward and then resets to zero upon each block arrival. The following theorem characterizes the stationary distribution of $t$ in the EO model.

\begin{prop}\label{prop: stationary distribution endogenous operation model}
    In the EO model, there is a unique stationary distribution of $t$. Furthermore, it has a probability density function, denoted $\psi$, which is given by
    \begin{equation}
    \psi(t; t^*) = \begin{cases}
        \dfrac{\lambda t^*}{1 + \lambda t^*} & \text{ if } t \in [0, t^*) \vspace{0.5em}\\
        \dfrac{\lambda e^{-\lambda (t - t^*)}}{1 + \lambda t^*} & \text{ if } t \in [t^*, +\infty)
    \end{cases},
    \end{equation}
    where $\psi(t;t^*)$ is the density of $t$ when the miner's threshold time is $t^*$.
\end{prop}

Similar to the UC model, the probability density decays at rate $\lambda$ while miners are operating. However, a key difference is that no blocks arrive when $t < t^*$, causing the probability density to remain constant during this period. As a result, as $t^*$ increases, the stationary distribution assigns more weight to larger values of $t$.

Next, we evaluate flow surplus and cost. Each user receives a utility of one when their transactions are validated. Since Lemma~\ref{lem: tE in 0 K} ensures that the equilibrium threshold time $t^E$ is in $[0, K)$ for any $y \ge 0$, we consider threshold time $t^*$ that lies in $[0, K)$. When the threshold time is $t^*$, the equilibrium probability that their transactions are validated is $W(0, t - t^*; t^*)$. Accordingly, the flow surplus originated from the users entering the market at time $t$ is as follows:
\begin{equation}
     W(0, t^* - t; t^*) = \begin{cases}
        1 - e^{-\lambda (K - (t^* - t))} & \text{ if } t < t^*\\
        1 - e^{-\lambda K} & \text{ if } t \ge t^*
    \end{cases}.
\end{equation}
On the other hand, when the miner works, the miner pays a flow operation cost of $c$.

We define social welfare, denoted by $SW(t^*)$, as the expectation of the flow surplus minus the flow operation cost where the expectation is taken according to the stationary distribution and $t^*$ is the threshold time:
\begin{equation}
    SW(t^*) =  \int_0^\infty (W(0, t^* - t; t^*) - c \textbf{1}\left\{t \ge t^*\right\}) \psi(t; t^*) dt.
\end{equation}
The function $SW$ is depicted in Figure~\ref{fig: EO SW}. We analyze the (socially) \emph{efficient threshold time}, denoted $t^O$, that maximizes the social welfare, along with the properties of the block reward $y$ that induces $t^O$. When $\lambda K \le c$, the social benefit from operation never exceeds the cost, implying that miners should never operate, or $t^O = + \infty$. The following theorem characterizes $SW$ and $t^O$ for the case of $\lambda K > c$.

\begin{prop}\label{prop: bar SW and tO}
    In the EO model, when $\lambda K > c$, $SW$ is uniquely maximized by $t^O$, defined as the unique solution of $t^*$ in the following equation:
    \begin{equation}\label{eq: efficient t star}
        t^* e^{-\lambda (K - t^*)} = \frac{c}{\lambda}.
    \end{equation}
\end{prop}

\begin{figure}[tb]
    \centering
    \begin{minipage}[b]{0.475\textwidth}
        \centering
        \includegraphics[width=\textwidth]{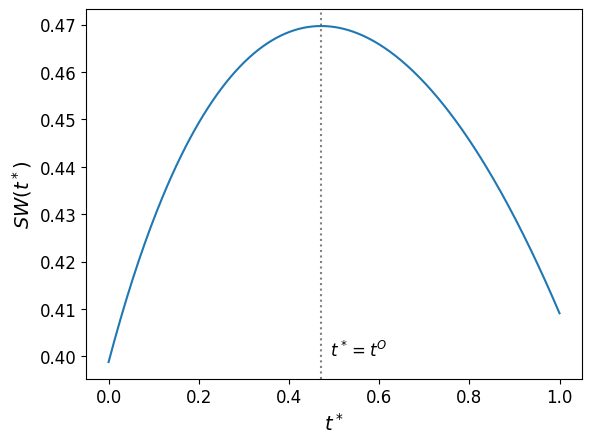}
        \subcaption{Social welfare $SW$}
        \label{fig: EO SW}
    \end{minipage}
    \hfill
    \begin{minipage}[b]{0.475\textwidth}  
        \centering 
        \includegraphics[width=0.98\textwidth]{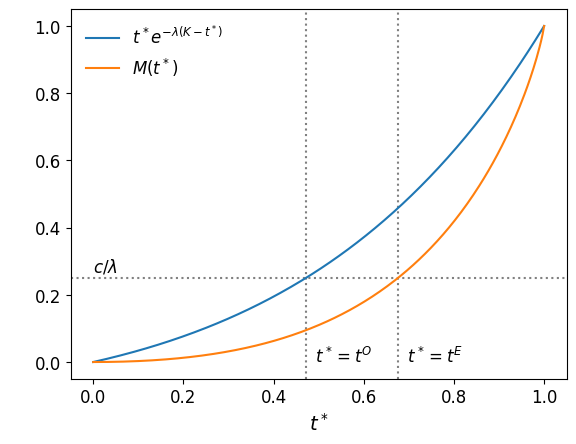}
        \subcaption{$t^E$ vs $t^O$}
        \label{fig: EO tO determination}
    \end{minipage}
    \caption{(a) Social welfare $SW$ and (b) the efficient and equilibrium threshold times, $t^O$ and $t^E$.}
    \label{fig: social welfare and tO}
    \begin{center}\footnotesize
        Parameters: $\lambda = 1.2$, $K = 1$, $c = 0.3$, $\eta = 0$.
    \end{center}
\end{figure}

In the following, we compare \eqref{eq: efficient t star} with
\begin{equation}\label{eq: y zero t star}
    M^*(t^*) = \frac{c}{\lambda},
\end{equation}
which characterizes the equilibrium threshold time given $y = 0$.

\begin{lem}\label{lem: to tstar comparison}
    In the EO model,  we have
        $t^* e^{-\lambda (K - t^*)} < M^*(t^*)$ for any $t^* \in (0, K)$.
\end{lem}

The left-hand side of \eqref{eq: efficient t star} represents the welfare received by all agents in the model, including users, when the miner starts working at time $t^*$, whereas the left-hand side of \eqref{eq: y zero t star} represents the surplus received only by the miner. Since the former considers the welfare of both miners and users, the value of the former is always greater than that of the latter, as shown in Figure~\ref{fig: EO tO determination}.

By equating these left-hand values with the normalized flow cost on the right-hand side, we can identify the efficient threshold time $t^O$ and the equilibrium threshold time $t^E(y)$, respectively. Let $y^O:=t^O e^{-\lambda (K - t^O)} - M^*(t^O)$. By Lemma~\ref{lem: to tstar comparison}, $y^O > 0$.

\begin{thm}
    In the EO model, the following claims hold.
    \begin{enumerate}
        \item The efficient threshold time is strictly smaller than the equilibrium threshold time given $y = 0$, i.e., $t^O < t^E(0)$.
        \item The equilibrium threshold time given the block reward $y^O$ is the efficient threshold time, i.e., $t^E(y^O) = t^O$. Furthermore, $y^O$ is the unique amount of the block reward that implements the efficient threshold time, i.e., $t^E(y) \neq t^O$ for any $y \neq y^O$.
    \end{enumerate}
\end{thm}

In equilibrium, when the block reward is zero, the miner decides whether to operate by comparing the total fees paid by users with their operational costs. This decision does not internalize the benefits that the users receive, leading to a socially inefficient outcome. To incentivize the miner to maximize social welfare, the optimal block reward $y^O$ should be determined by estimating the users' benefits at the socially efficient threshold time $t^O$, assuming that the miner begins operating from this point onward.

\section{Patient Users}\label{sec: patient user summary}

Thus far, we have considered a scenario where the game ends once a block arrives, and any users whose transactions were not validated receive no payoff. Although this model uncovers the basic incentives of users and miners in the Bitcoin market, in reality, users may continue to wait when their transactions are not validated in the first block arrival, expecting validation by future blocks. In this case, the transaction fee chosen by users is not only determined by the probability that their transaction is included in the first block arrival, but also by the amount of discounting until the time their transaction is validated (which may not be at the first block arrival).

In Appendix~\ref{sec: appendix patient users}, we analyze the \emph{patient user model}, where users continue to wait until their transaction is processed. In this model, it is difficult to analytically derive the relationship between the fee set by users and the time it takes for their transaction to be processed. This is because the $W$ function should be replaced with a much more complicated function without a closed form, and it is not immediate if the solution to the first-order condition satisfies global optimality. However, the behavior of the patient user model, derived numerically, qualitatively resembles that of the models we have discussed, which we solved analytically.

\section{Conclusion}\label{sec: conclusion}

This paper develops a dynamic model of user and miner behavior in the Bitcoin market, capturing how short-term congestion influences the entire market's equilibrium behavior. We characterize the equilibrium bid function, showing how users adjust fees based on short-term congestion as well as the average throughput. We also highlight how miner incentives fluctuate, leading to temporary suspensions in mining activity for extended periods when transaction volumes are low, thereby reducing social welfare. Finally, we characterize the amount of the block reward to resolve this problem and induce socially optimal miner operations.

Our findings have broader implications beyond Bitcoin. Randomness in block arrivals (or other factors that temporarily alter the balance of demand and supply, e.g., user arrivals) is not unique to Bitcoin but is also prevalent in other cryptocurrencies that employ vastly different fee mechanisms and consensus mechanisms. Furthermore, beyond cryptocurrencies, various types of consumer-driven fee-setting markets are developed and being used in today's economy. By understanding the dynamics of congestion and strategic fee adjustments, we can develop more efficient mechanisms to stabilize the market performance. In the context of Bitcoin, our results suggest that the adjustment to block reward, which is a modest policy intervention, has a significant impact, effectively resolving the problem of miners' temporary suspensions.

\bibliographystyle{ecta}
\bibliography{reference}

\appendix

\part*{Appendix}

\section{Patient Users}\label{sec: appendix patient users}

\subsection{Model}

In this section, we formulate and analyze a model with patient users. The differences between the patient user model and the impatient user model, analyzed in detail in this paper, are twofold. First, in the patient user model, transactions that are not validated by the arriving block remain in the pool. In other words, if the number of transactions in the pool is $t$ at the time of block arrival, the $\min\{t, K\}$ units of transactions with the highest fees are processed, and the remaining $t - \min\{t, K\}$ transactions stay in the pool, allowing the game to continue. Second, the user's utility function incorporates a time discounting factor, where users derive higher utility the sooner their transaction is processed. Specifically, the user's ex post utility is
\begin{equation}\label{eq: impatient user model ex post payoff}
    e^{- \rho (\tau' - \tau)}(1 - b),
\end{equation}
where $\rho$ is the user's time discount factor, $\tau$ is the time of the user arrival, $\tau'$ is the time when the transaction is validated, and $b$ is the fee set by the user.
When the user's transaction request is never validated, her payoff is set to be $0$.

A \emph{bid function} $\beta: \mathbb{R}_+ \to \mathbb{R}_+$ is a measurable function that assigns a bid to each mass of unprocessed transaction requests. We consider situations where all agents facing the same size of the mass of unprocessed transactions, say $t$, bid the same price, say $\beta(t)$. 
Since the mass of transaction requests increases continuously by one unit per unit of time, it shares similar properties with time $t$ in the impatient user model. Although there is a difference where the mass of transaction requests returns from $t$ to $t - K$ each time a block arrives, following the convention of the impatient user model, we refer to $t$ as ``time.''

We further focus on bid functions satisfying the following properties:
\begin{enumerate}[(i)]
    \item $\beta(t) \in [0, 1)$ for all $t \in \mathbb{R}$.
    \item $\beta$ is strictly increasing.
    \item $\beta$ is continuous.
    \item $\beta(0) = 0$.
    \item $\lim_{t \to \infty} \beta(t) = 1$.
\end{enumerate}
For the impatient-user model, these are basic properties that arise as necessary conditions $\beta$ should satisfy (Theorems~\ref{thm: basic property of eqm one-shot always operating} and \ref{thm: basic property of eqm Model 2}). Assuming that the bid function satisfies these properties, we describe the behavior of the economy and define the equilibrium.

When all users follow an increasing bid function $\beta$, transactions are validated in descending order of the user arrival time. Therefore, if a block arrives at time $t + K$, the transactions that will be validated are those that arrived within the time interval $(t, t + K]$, as if the ``time'' is turned back to $t$. Consequently, as long as the current ``time'' $t$ is identical, the user faces the identical distribution of bids in the pool and identical probability distributions over future events. Therefore, the distribution of bids in the pool follows a Markov process with $t$ being a state variable.

The time $t$ in the impatient user model, which acts as the state variable, increases by $1$ unit per unit of (physical) time and follows a stochastic process where blocks arrive at a Poisson rate $\lambda$. Each time a block arrives, the current time transitions from $t$ to $t - \min\{t, K\}$. If a user reports time $\hat{t}$, meaning they bid $\beta(\hat{t})$, their transaction is validated at the first moment after their arrival when at least one block has arrived and the current time becomes no more than $\hat{t}$. Each user determines their bid based on the expected time discount associated with this stochastic process.

\subsection{Equilibrium}

Let $\tilde{W}(s) \coloneqq \mathbb{E}[e^{-\rho T}|s]$ be a function representing expected discount given $s$. Time-$t$ user's problem is
\begin{equation}\label{eq: impatient user payoff function}
    \max_{s \in [-t, +\infty)} \tilde{W}(s)(1 - \beta(t - s)),
\end{equation}
which is equivalent to the user's problem in the impatient user model except that we replace the $W$ function with the $\tilde{W}$ function. We say that a bid function $\beta$ is an \emph{equilibrium} if for all $t \in \mathbb{R}_+$, $s = 0$ solves \eqref{eq: impatient user payoff function}.

The $\tilde{W}$ function satisfies the following differential equation.

\begin{thm}\label{thm: differential equation W-tilde satisfies}
    In the patient user model, the $\tilde{W}$ function satisfies
    \begin{equation}\label{eq:diff_W-tilde}
        \tilde{W}'(s) = (\rho + \lambda)\tilde{W}(s) - \lambda \tilde{W}^*(s - K),
    \end{equation}
    where
    \begin{equation}
        \tilde{W}^*(s) = \begin{cases}
            \tilde{W}(s) & \text{ if } s \ge 0;\\
            1 & \text{ otherwise.}
        \end{cases}
    \end{equation}
    Furthermore, $\lim_{s \to - \infty}\tilde{W}(s) = \lambda/(\rho + \lambda)$, and $\lim_{s \to + \infty}\tilde{W}(s) = 0$.
\end{thm}

By the Picard-Lindel\"{o}f theorem, together with a boundary condition specifying the values of $\tilde{W}(K)$, the differential equation \eqref{eq:diff_W-tilde} uniquely determines the shape of the $\tilde{W}$ function. While a formal proof has not been established, we conjecture that the derived $\tilde{W}$ function is the desired one if it also satisfies the two limit conditions, $\tilde{W}(s) \to \lambda/(\rho + \lambda)$ as $s \to - \infty$ and $\tilde{W}(s) \to 0$ as $s \to + \infty$. However, both boundary conditions specify the values at the limits, making them intractable. Even slight changes to the initial values $\tilde{W}(0)$ or $\tilde{W}(K)$ can cause substantial variations in $\tilde{W}(s)$ for large values of $s$. Consequently, deriving the entire shape of the $\tilde{W}$ function from the differential equation is challenging. Moreover, the $\tilde{W}$ function itself exhibits kinks at points where $s = nK$ for some integer $n$, further complicating the derivation of the function's shape in a closed form.

We numerically derive the $\tilde{W}$ function. Since block arrivals follow a Poisson process with intensity $\lambda$, the time $B$ until the next block arrives is exponentially distributed with rate $\lambda$. Utilizing this property, we generated many exponentially distributed random variables $B(n)$ and updated the original time (or pool size) $t$ to $t + B(n) - K$ as the new time after the $n$th block arrives. By recording the value of $\sum B(n)$ at the moment when the time $t$ first becomes smaller than $s$, we calculate the distribution of the duration $T$ until a transaction request is validated, given $s$. Based on this distribution of $T$, we then derived the shape of the $W$ function according to the original definition.

\begin{figure}
    \centering
    \includegraphics[width=0.5\linewidth]{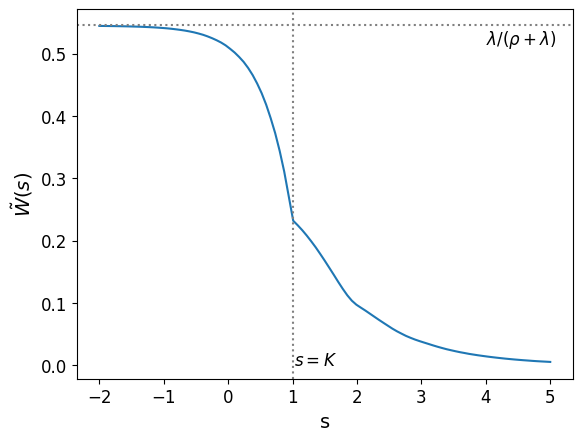}
    \caption{The $\tilde{W}$ function.}
    \begin{center}\footnotesize
        Parameters: $\lambda = 1.2$, $K = 1$, $\rho = 1$.
    \end{center}
    \label{fig: patient_W}
\end{figure}

Figure~\ref{fig: patient_W} depicts the shape of the $\tilde{W}$ function. $\tilde{W}$ is decreasing in $s$ because the bid $\beta(t - s)$ becomes weaker as $s$ increases, compared with the on-path bid of a currently arriving user, $\beta(t)$. As $s \to -\infty$, $\tilde{W}(s) \to \lambda/(\rho + \lambda)$, which is the expected discount factor given that the transaction is validated at the next block arrival---this assumption is satisfied only when at least one block arrives within $s + K$ unit of time, but the probability of such an event approaches one as $s$ goes to $-\infty$. Unlike the impatient-user model, $W(s)$ does not become zero at $s = K$ because the user's transaction may be validated after multiple block arrivals. Still, as $s \to \infty$, we have $\tilde{W}(s) \to 0$ because extremely large $s$ means that it will take a very long time for the transaction to be validated, and thus the discounted payoff will be very small.

Since we do not have an analytical solution for the $\tilde{W}$ function, it remains unclear whether the user's best response to a given bid function $\beta$ can be characterized by the first-order condition. However, we will proceed under the assumption that the first-order condition provides a sufficient condition for optimality. When all other agents follow the bid function $\beta$, the utility maximization problem for a time-$t$ agent is given by \eqref{eq: impatient user payoff function}, and its first-order condition is as follows:
\begin{equation}\label{eq: W-tilde FOC}
    \tilde{W}'(0) (1 - \beta(t)) + \tilde{W}(0)\beta'(t) = 0.
\end{equation}
Solving this, we obtain a candidate equilibrium bid function $\tilde{\beta}^E$ as follows:

\begin{thm}
    In the patient user model, suppose that for all $t \in \mathbb{R}_+$, \eqref{eq: W-tilde FOC} is a sufficient condition for a time-$t$ user to bid $\beta(t)$ as a unique best response. Then, the bid function specified by
    \begin{equation}\label{eq: impatient bid function}
        \beta(t) = \tilde{\beta}^E(t) = 1 - e^{\frac{\tilde{W}'(0)}{\tilde{W}(0)}t}.
    \end{equation}
    is the unique equilibrium.
\end{thm}

This functional form is the same as \eqref{eq: eqm bid function one-shot always operating}, except that $W$ is replaced with $\tilde{W}$, whose closed-form representation is unknown. In Figure~\ref{fig: patient optimization}, we illustrate a time-$t$ user's payoff from each $s$, given that other users follow the bid function specified in \eqref{eq: impatient bid function}. Based on these simulation results, the first-order condition appears to be sufficient for user payoff maximization as intended. Therefore, we proceed by regarding the bid function specified by \eqref{eq: impatient bid function} as the equilibrium and continue the discussion accordingly.

\begin{figure}[tb]
    \centering
    \begin{minipage}[t]{0.475\textwidth}
        \centering
        \includegraphics[width=\textwidth]{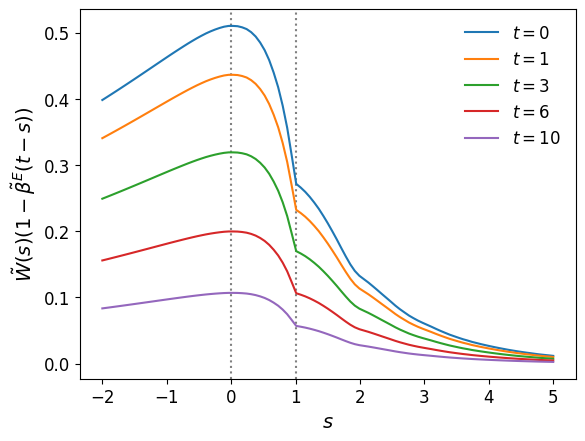}
        \subcaption{Time-$t$ user's payoff function, $\tilde{W}(s)(1 - \tilde{\beta}^E(t - s))$, across various $s$}
        \label{fig: patient optimization}
    \end{minipage}
    \hfill
    \begin{minipage}[t]{0.475\textwidth}  
        \centering 
        \includegraphics[width=\textwidth]{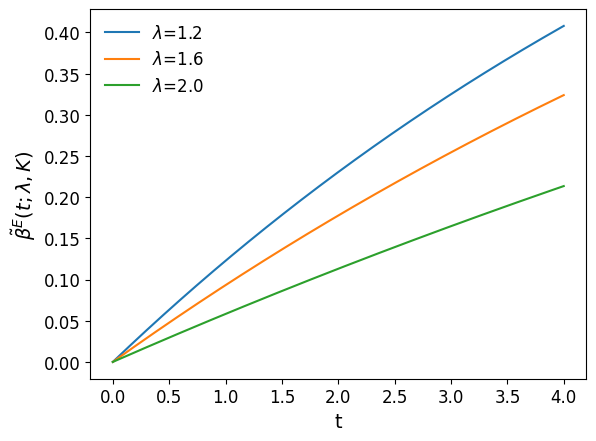}
        \subcaption{The equilibrium bid function $\tilde{\beta}^E(t)$}
        \label{fig: patient beta}
    \end{minipage}
    \caption{The user's payoff function and the equilibrium bid function.}
    \begin{center}\footnotesize
        Parameters: $\lambda = 1.2$ (unless explicitly specified), $K = 1$, $\rho = 1$.
    \end{center}
    \label{fig: patient optimization and beta}
\end{figure}

Figure~\ref{fig: patient beta} plots the equilibrium bid function $\tilde{\beta}^E$, corresponding to Figure~\ref{fig: user competition variation lambda} in the impatient user model. While the $W$ function is replaced by $\tilde{W}$, resulting in a quantitatively different equilibrium, there is no significant difference in the overall shape of the equilibrium bid function between the two models---the equilibrium bid function is increasing in $t$, $\lambda$, and $K$, and concave in $t$. Therefore, the policy implications regarding how to set parameters like $\lambda$, $K$, and $y$ should align with those obtained from the simple and tractable impatient user model, as changes in these parameters would yield qualitatively similar effects on the behavior of the economy.

\section{Proofs}\label{sec: proofs}

\subsection{Proof of Theorem~\ref{thm: basic property of eqm one-shot always operating}}\label{subsec: basic property of eqm one-shot always operating}

The proposition is implied by Theorem~\ref{thm: basic property of eqm Model 2}.

\subsection{Proof of Theorem~\ref{thm: equilibrium closed form one-shot always operating}}\label{subsec: equilibrium closed form one-shot always operating}

The proposition is implied by Theorem~\ref{thm: equilibrium closed form model 2}.

\subsection{Proof of Proposition~\ref{prop: stationary distribution}}

\begin{proof}
    If $\Psi$ is a stationary distribution, it must satisfy the following equation for all $t \in (0, + \infty)$ and $\Delta t \in (0, t)$
    \begin{equation}\label{eq: user competition stationary 4}
        \Psi(t) = 1 - e^{- \lambda \Delta t} + e^{- \lambda \Delta t} \Psi(t - \Delta t).
    \end{equation}
    If $\Psi(0) > 0$ were the case, then for any $t \in (0, + \infty)$,   a strictly positive probability mass of $e^{- \lambda t} \Psi(0)$ must remain because the probability that no block arrives between $0$ and $t$ is strictly positive ($e^{- \lambda t}$). However, since a probability distribution cannot have a positive mass at all points in a continuous interval, this is a contradiction. Therefore, $\Psi(0) = 0$ must hold. Furthermore, it is straightforward to derive the following differential equation from \eqref{eq: user competition stationary 4}:
    \begin{equation}
        \Psi'(t) = \lambda (1 - \Psi(t)).
    \end{equation}
    By solving this differential equation with the initial condition of $\Psi(0) = 0$, we obtain
    \begin{equation}
        \Psi(t) = 1 - e^{- \lambda t}
    \end{equation}
    as a unique candidate of the stationary distribution. Furthermore, this distribution satisfies \eqref{eq: user competition stationary 4}, implying that it is indeed a stationary distribution.
\end{proof}

\subsection{Proof of Theorem~\ref{thm: basic property of eqm Model 2}}\label{subsec: basic property of eqm Model 2}

\begin{proof} 
\textbf{(i)} Towards a contradiction, suppose that there exists $t$ such that $\beta(t) \ge 1$. Let $t^* = \inf\{t: \beta(t) \ge 1\}$. Then, for any $\bar{\delta} > 0$, there exists $\delta \in [0, \bar{\delta})$ such that $\beta(t^* + \delta) \ge 1$. If time-$(t^* + \delta)$ user bids $\beta(t^* + \delta) \ge 1$, then her payoff is nonpositive. Let $b^*$ be the $(K/2)$th highest bid in time $t^*$, i.e.,
\begin{equation}
    b^* = \inf_{b' \in [0, 1]} \left\{|\{t \in [0, t^*]: \beta(t) > b'\}| \le \frac{K}{2}\right\}.
\end{equation}
Since $\beta(t) < 1$ for all $t < t^*$, $b^* < 1$ holds. Furthermore, this definition of $b^*$ implies that, for $\bar{\delta} < K/2$,
the mass of users with a bid above $b$, $L(t+\delta,b,\beta)$, is strictly less than $K/2+K/2$, or $K$. Thus, since  $L(\tilde{t},b,\beta)$ is continuous in $\tilde{t}$ by its definition \eqref{eq: model 1, L}, we have $\bar{t}(b,\beta)>t+\delta$ by the definition of $\bar{t}$,  \eqref{eq: model 1, bar t} (that is, the bid $b$ can be validated in some time interval strictly after $t$). This implies $\pi(t,b;\beta,\sigma)>0$ because $A(\tilde{t},\sigma)$ is strictly increasing in $\tilde{t}$, which means that by bidding $b \in (b^*, 1)$, the probability that time-$(t^* + \delta)$ user's bid is validated is positive because $\eta>0$. Since she obtains a positive payoff of $(1 - b)$ when it is validated, the time-$(t^* + \delta)$ user's payoff from bidding $b$ is strictly positive. Accordingly, this user can improve her payoff from a deviation, and therefore, $\beta$
does not satisfy the user's optimality condition of Model~2 given  $\sigma$.

\textbf{(ii)} Fix $t$ arbitrarily. First, we show that $\beta$ is nondecreasing, i.e., $\beta(t + \Delta) \ge \beta(t)$ for any $\Delta > 0$. For time-$(t + \Delta)$ user, bidding $b$ such that $\bar{t}(b, \beta) \le t + \Delta$ is suboptimal because such bid results in the payoff of 0 while bidding $b^*$ as specified in the proof of part (i) with $t^*$ being replaced with $t+\Delta$ would provide a strictly positive payoff to the user, which follows essentially the same argument as in the proof of part (i). Thus, it suffices to consider $b$ such that $\bar{t}(b, \beta) > t + \Delta$. For such $b$, we have
\begin{align}
    \pi(t + \Delta, b; \beta, \sigma) & = 1 - e^{- (A(\bar{t}(b, \beta); \sigma) - A(t + \Delta; \sigma))} \\
    & = (1-e^{-(A(\bar{t}(b,\beta); \sigma)-A(t; \sigma))})+(1-e^{A(t+\Delta; \sigma)-A(t; \sigma)})e^{-(A(\bar{t}(b,\beta); \sigma)-A(t; \sigma))}\\
    & = \pi(t, b; \beta, \sigma) + (1-e^{A(t+\Delta; \sigma)-A(t; \sigma)})e^{-(A(\bar{t}(b,\beta); \sigma)-A(t; \sigma))},
\end{align}
where we used the fact that $\bar{t}(b, \beta) > t + \Delta > t$ for the last equality.
Accordingly,
\begin{equation}
        \pi(t + \Delta, b; \beta, \sigma)(1 - b) = \pi(t, b; \beta, \sigma)(1 - b) + (1-e^{A(t+\Delta; \sigma)-A(t; \sigma)})e^{-(A(\bar{t}(b,\beta); \sigma)-A(t; \sigma))}(1 - b).\label{eq: decomposed time t user's payoff 2}
\end{equation}
The left-hand side of \eqref{eq: decomposed time t user's payoff 2} is the time-$(t + \Delta)$ user's payoff from bidding $b$. The first term of the right-hand side of \eqref{eq: decomposed time t user's payoff 2} is the time-$t$ user's payoff from bidding $b$, and therefore, $\beta(t)$ is a maximizer of the first term. Furthermore, since (i) $1-e^{A(t+\Delta; \sigma)-A(t; \sigma)} < 0$ because $\eta>0$, (ii) $\bar{t}$ is nondecreasing in $b$, and (iii) $(1 - b)$ is decreasing in $b$, the second term of the right-hand side of \eqref{eq: decomposed time t user's payoff 2} is increasing in $b$. Therefore, any bid strictly below $\beta(t)$ provides the time-$(t+\Delta)$ user with a strictly lower payoff than $\beta(t)$ does. This implies that $\beta(t + \Delta) \ge \beta(t)$.

Now we show that $\beta$ is strictly increasing. Towards a contradiction, suppose otherwise. Then, since $\beta$ is nondecreasing, there exist $t', t'' \in [0, + \infty)$ such that $t' < t''$ and $\beta(t) = \beta(\hat{t})$ for all $t, \hat{t} \in (t', t'')$. Then, for any $t \in (t', t'')$ and any $\epsilon > 0$, we have $\bar{t}(\beta(t) + \epsilon, \beta) \ge \bar{t}(\beta(t), \beta) + t'' - t'$. This implies $\bar{t}(\beta(t) + \epsilon, \beta) > \bar{t}(\beta(t), \beta) $.
Also, note that $\bar{t}(\beta(t) + \epsilon, \beta) \geq t+K$ because $\beta$ is nondecreasing, and thus 
\begin{align}
    & \pi(t, \beta(t) + \epsilon; \beta, \sigma)- \pi(t, \beta(t); \beta, \sigma)\\
    &= (1-e^{- [A(\bar{t}(\beta(t) + \epsilon, \beta);\sigma) -A(t,\sigma)]^+})-(1-e^{- [A(\bar{t}(\beta(t), \beta);\sigma) -A(t,\sigma)]^+})\\
    & = \begin{cases} 1-e^{- [A(\bar{t}(\beta(t) + \epsilon, \beta);\sigma) -A(t,\sigma)]}\geq 1-e^{-\lambda \eta K}>0 
    &\text{ if }\;\bar{t}(\beta(t)  , \beta) -t\leq 0\\
    e^{A(t,\sigma)}(e^{A(\bar{t}(\beta(t)+\epsilon,\beta),\sigma)} - e^{A(\bar{t}(\beta(t) , \beta),\sigma)}) \\
    \hspace{2em} \ge e^{A(t,\sigma)}e^{A(\bar{t}(\beta(t) , \beta),\sigma)}e^{\lambda \eta(t''-t')} >e^{\lambda \eta(t''-t')}>0 &
    \text{ if }\;\bar{t}(\beta(t)  , \beta) -t>0
\end{cases}.
\end{align}

Accordingly, $\lim_{\epsilon \to 0}\pi(t, \beta(t) + \epsilon; \beta, \sigma) > \pi(t, \beta(t); \beta, \sigma)$. Hence,
\begin{align}
    &\lim_{\epsilon \to 0}\pi(t, \beta(t) + \epsilon; \beta, \sigma)(1 - (\beta(t) + \epsilon))\\
    =& \lim_{\epsilon \to 0}\pi(t, \beta(t) + \epsilon; \beta, \sigma)(1 - \beta(t))\\
    >& \pi(t, \beta(t); \beta, \sigma)(1 - \beta(t)),
\end{align}
where the inequality holds because of our earlier conclusion that $\beta(t)<1$ for every $t$.
Hence, the time-$t$ user can increase her payoff by bidding $\beta(t) + \epsilon$ with sufficiently small $\epsilon > 0$, contradicting the assumption that $\beta$ satisfies the user's optimality condition given $\sigma$. Accordingly, $\beta$ is strictly increasing.

\textbf{(iii)} Suppose that $\beta$ is not continuous. Then, there exist $t \in \mathbb{R}_+$ and $\epsilon > 0$ such that for all $\delta > 0$, there exists $\hat{t}$ such that $|t - \hat{t}| < \delta$ and $|\beta(t) - \beta(\hat{t})| > \epsilon$.
Let $t' = \min\{t, \hat{t}\}$ and $t'' = \max\{t, \hat{t}\}$. Then, $t' < t'' < t'+\delta$, Furthermore, since $\beta$ is increasing, we have $\beta(t') + \epsilon < \beta(t'')$. Consider time-$t''$ user's problem. If this user bids $\beta(t'')$, then her payoff is $\pi(t'', \beta(t''); \beta, \sigma) (1 - \beta(t''))$. If she deviates to $\beta(t')$, her payoff is $\pi(t'', \beta(t'); \beta, \sigma)(1 - \beta(t'))$. Accordingly, her gain from deviation is
\begin{align}\label{eq:beta_cont}
    & \pi(t'', \beta(t'); \beta, \sigma)(1 - \beta(t')) - \pi(t'', \beta(t''); \beta, \sigma)(1 - \beta(t''))\\
    >& -(\pi(t'', \beta(t''); \beta, \sigma) - \pi(t'', \beta(t'); \beta, \sigma))(1 - \beta(t')) + \epsilon \pi(t'', \beta(t''); \beta, \sigma).
\end{align}
Also, strict increasingness of $\beta$  implies $\bar{t}(\beta(t), \beta) = t + K$ for all $t$. Accordingly,
\begin{equation}
    \pi(t, \beta(t - s); \beta, \sigma) = 1 - e^{-[A(t - s + K) - A(t)]}.
\end{equation}
Hence, in particular, $\pi(t, \beta(\cdot); \beta, \sigma)$ is a continuous  function. Thus, as $\delta \to 0$, the former term of the right-hand side of \eqref{eq:beta_cont} converges to $0$.
This implies that the right-hand side of \eqref{eq:beta_cont} is strictly positive for sufficiently small $\delta>0$ because $\eta>0$ implies $\pi(t'', \beta(t''), \beta)>0$, which in turn implies that
time-$t''$ user's gain from deviation is strictly positive. This contradicts our assumption that $\beta$ satisfies the user's optimality condition given $\sigma$. Hence, $\beta$ must be continuous.

\textbf{(iv)} Suppose that $\beta(0) > 0$. Since $\beta$ is strictly increasing, $\beta(t) > \beta(0) > 0$ for all $t > 0$. By bidding $0$ instead of $\beta(0)$, a time-$0$ user can reduce her payment conditional on her bid being validated, without changing the distribution of times at which her transaction request is validated. Since $\eta>0$ implies that the probability that her bid is validated is strictly positive, this implies that $\beta$ does not satisfy the user's optimality condition given $\sigma$, a contradiction. Hence, $\beta(0)=0$ must hold.

\textbf{(v)} Since $\beta$ is increasing and bounded by $1$, there exists $b^*<\infty$ such that $\lim_{s\to \infty}\beta (s) = b^*$. Towards contradiction, suppose that $b^* < 1$. Take $\epsilon \in (0, b^*)$ arbitrarily and define $t < \infty$ by $\beta(t) = b^* - \epsilon$. Such $t$ exists because $\beta(0)=0$ and $\beta$ is continuous. If time-$t$ user bids $b^*$ instead of $\beta(t)$, her payoff is $1 - b^*$ because $\eta>0$ and $\bar{t}(b^*, \beta) = + \infty$, where the latter holds as $\beta(\tilde{t})<b^*$ for all $\tilde{t}$. If time-$t$ user follows $\beta$ and bids $\beta(t) = b^* - \epsilon$, her payoff is $\pi(t, \beta(t); \beta, \sigma)(1 - (b^* - \epsilon))$.
Since strict increasingness of $\beta$ implies $\bar{t}(\beta(t), \beta) = t + K$, we have $\pi(t, \beta(t); \beta, \sigma) = 1 - e^{-(A(t + K) - A(t))} \le 1 - e^{-\lambda K} < 1$. Thus, for sufficiently small $\epsilon > 0$, we have
\begin{equation}
    1 - b^* > (1 - e^{-\lambda K})(1 - (b^* - \epsilon)) \ge \pi(t, \beta(t); \beta, \sigma)(1 - (b^* - \epsilon)),
\end{equation}
implying that time-$t$ user obtains a larger payoff by bidding $b^*$ instead of $\beta(t) = b^* - \epsilon$. This contradicts the assumption that $\beta$ satisfies the user's optimality condition. Hence, $b^* = 1$, as desired.
\end{proof}

\subsection{Proof of Lemma~\ref{lem: optimality of threshold strategy}}

\begin{proof}
By part (ii) of Theorem~\ref{thm: basic property of eqm Model 2}, if $\beta$ satisfies the user's optimality condition, $\beta$ is strictly increasing. Thus, upon finding a new block, the miner validates mass $K$ of the latest transaction requests. Accordingly,
\begin{equation}
    M(t; \beta) = \int_{[t-K]^+}^t \beta(s) ds.
\end{equation}
Since $\beta$ is increasing, $M$ is also strictly increasing in $t$. Accordingly, in equilibrium, $\sigma$ is a threshold strategy.
\end{proof}

\subsection{Proof of Theorem~\ref{thm: equilibrium closed form model 2}}\label{subsec: equilibrium closed form model 2}

\begin{proof} 

\textbf{(Sufficiency)} First, we prove that part (i) implies the miner's optimality. The miner's optimality condition requires that the miner operates when $\lambda(M(t, \beta) + y) > c$ and suspends when $\lambda(M(t, \beta) + y) < c$. In any equilibrium $(\beta, t^*)$, $M(t, \beta) \in [0, K]$. Accordingly, $\lambda y > c$ implies $\lambda(M(t, \beta) + y) > c$ for all $t$, and $\lambda (y + K) < c$ implies $\lambda(M(t, \beta) + y) < c$ for all $t$. Furthermore, since $M(t, \beta)$ is nondecreasing in $t$, if $t$ solves $\eqref{eq: miner's threshold model 2}$, then $\lambda(M(t', \beta) + y) \ge c$ for all $t' > 0$ and $\lambda(M(t', \beta) + y) \le c$ for all $t' < 0$. Accordingly, part (i) implies the miner's optimality.

Next, we show that part (ii) implies the user's optimality.

\begin{lem}\label{lem: V single peaked, one-shot, Model 2}
In the EO model, for all $t \in [0, +\infty)$, $W(s, t^* - t)(1- \beta^E(t-s, t^*))$ is increasing in $s$ for $s \in (-\infty, 0]$, decreasing in $s$ for $s \in [0, K]$, and fixed to $0$ for all $s \in [K, \infty)$.
\end{lem}
\begin{proof}
The final part straightforwardly follows from $W(s, t^* - t) = 0$ for all $s \in [K, \infty)$.
Since $\beta^E$ and $W$ can be shown to be continuous by inspection, $W(s, t^* - t)(1- \beta^E(t-s, t^*))$ is continuous at $s=K$. Therefore, it suffices to show that  $W(s, t^* - t)(1- \beta^E(t-s, t^*))$ is increasing in $s$ for $s \in (-\infty, 0]$ and decreasing in $s$ for $s \in [0, K)$.
To see this, fix $t\in[0,\infty)$ and $s\in(-\infty,K)$. Note that we must have $W(s, t^* - t), W(0, t^* - t)>0$. Observe that
\begin{align}
    &\frac{d}{ds}\left[W(s, t^* - t)(1 - \beta^E(t-s, t^*)\right]\\
    = & W_1(s, t^* - t)(1 - \beta^E(t-s, t^*)) + W(s, t^* - t)\beta_1^E(t-s, t^*)\\
    = & W_1(s, t^* - t)(1 - \beta^E(t-s, t^*)) - W(s, t^* - t)\frac{W_1(0, t^* - t)}{W(0, t^* - t)}(1 - \beta_1^E(t-s, t^*))\\
    = & \left(\frac{W_1(s, t^* - t)}{W(s, t^* - t)} - \frac{W_1(0, t^* - t)}{W(0, t^* - t)}\right)W(s, t^* - t) (1 - \beta^E(t-s, t^*)).%
\end{align}
For $s \in [0, K)$, we have
\begin{equation}
    \frac{W_1(s, t^* - t)}{W(s, t^* - t)} = \begin{cases}
        -\dfrac{\eta \lambda e^{- \eta \lambda (K - s)}}{1 - e^{- \eta \lambda (K - s)}} & \text{ if } t - s + K \le t^*,\\
        -\dfrac{\lambda e^{- \eta \lambda (t^* - t) - \lambda (K - s + t - t^*)}}{1 - e^{- \eta \lambda (t^* - t) - \lambda (K - s + t - t^*)}} & \text{ if } t \le t^* < t - s + K,\\
        -\dfrac{\lambda e^{- \lambda (K - s)}}{1 - e^{- \lambda (K - s)}} & \text{ if } t > t^*.
    \end{cases}
\end{equation}
By direct calculation, we can show that $W_1(\cdot, t^* - t)/W(\cdot, t^* - t)$ is strictly decreasing. This implies that $W_1(s, t^* - t)/W(s, t^* - t)>W_1(0, t^* - t)/W(0, t^* - t)$ for $s<0$ and $W_1(s, t^* - t)/W(s, t^* - t)<W_1(0, t^* - t)/W(0, t^* - t)$ for $s>0$. Note that $W(s, t^* - t)$ and $1 - \beta(t-s)$ is positive for all $t$ and $s \in (-\infty, K)$, and $W(s, t^* - t)$ is decreasing in $s$. Hence, $W(s, t^* - t)(1- \beta^E(t-s, t^*))$ is increasing in $s$ for $s \in (-\infty, 0]$ and decreasing in $s$ for $s \in [0, K]$, as desired.
\end{proof}

\noindent
\textbf{(Existence and Uniqueness)}
It suffices to show that, if $\lambda y < c < \lambda (K + y)$, there exists $t'$ such that
\begin{equation}\label{eq: existence proof 1}
    \lambda (M^*(t') + y) = c,
\end{equation}
where $M^*(t') = M(t'; \beta^E(\cdot, t'))$.
Since $M(0, \beta) = 0$ for all $\beta$, we have
\begin{equation}
    \lambda (M^*(0) + y) = \lambda y < c.
\end{equation}
Furthermore, it follows from $\beta^E(t, t') \to 1$ as $t \to \infty$ for all $t'$ and $\beta^E$ is increasing in $t'$ that for $\epsilon = K + y - c/\lambda$, there exists $T$ such that $M(T; \beta^E(\cdot, T)) > K - \epsilon$. For such $T$,
\begin{equation}
    \lambda (M^*(T) + y) > \lambda (K - \epsilon + y) = c.
\end{equation}
We can verify that $M^*$ is a continuous function by inspecting its closed functional form. By the intermediate value theorem, there exists $t'$ that satisfies \eqref{eq: existence proof 1}, as desired. Furthermore, by Theorem~\ref{thm: basic property of eqm Model 2} and Theorem~\ref{thm: EO model beta is increasing in threshold}, $\beta^E(t, t')$ is increasing in $t$ and nondecreasing in $t'$. Accordingly, $M^*$ is an increasing function, implying that there exists at most one $t'$ satisfying \eqref{eq: existence proof 1}.

\noindent
\textbf{(Necessity given $\eta > 0$)}
First, we show that $(\beta, t^*)$ is an equilibrium only if part (i) is satisfied. Since $\eta > 0$, $\beta$ could be an equilibrium only if it is increasing. Clearly, when $\lambda y > c$ but $t^* > 0$ and $\lambda (y + K) < c$ but $t^* < + \infty$ violates the miner's optimality condition. Consider the case of $\lambda y \le c \le \lambda (y + K)$. Since $\beta$ is increasing, $M(t; \beta)$ is also increasing in $t$. Furthermore, it is easy to see that $M(t; \beta)$ is continuous in $t$. Accordingly, there exists a unique $t$ satisfying $\lambda (M(t;\beta) + y) = c$, and $\lambda (M(t';\beta) + y) > c$ if $t' > t$ and $\lambda (M(t';\beta) + y) < c$ if $t' < t$.

Next, we show that the user's optimality is satisfied only if part (ii) is satisfied given $\eta > 0$.

\begin{lem}\label{lem: equilibrium necessity model 2}
In the EO model, if $\eta > 0$ and $(\beta, t^*)$ is an equilibrium, then $\beta = \beta^E(\cdot, t^*)$ must be the case.
\end{lem}

\begin{proof}
By direct calculation, we can show that $W(t, s; t^*)$
as defined in the statement of (ii) is differentiable at $s = 0$ for all $t \in \mathbb{R}_+\setminus\{t^* - K\}$.

By the definition of $W$, a necessary condition for $(\beta,t^*)$ to be an equilibrium is
\begin{equation}
    W(0, t^* - t) (1-\beta(t)) \ge W(\epsilon, t^* - t) (1 - \beta(t - \epsilon))
\end{equation}
for all $\epsilon > 0$ and for all $t \in \mathbb{R}_+\setminus\{t^* - K\}$. This inequality can be rearranged to
\begin{equation}
    \frac{\beta(t) - \beta(t-\epsilon)}{\epsilon} \le - \frac{1}{W(0, t^* - t)} (1 - \beta(t-\epsilon)) \frac{W(\epsilon, t^* - t) - W(0, t^* - t)}{\epsilon}.
\end{equation}
Since $\beta$ is continuous (by part (ii) of Theorem~\ref{thm: basic property of eqm Model 2}), $\beta(t - \epsilon) \to \beta(t)$ as $\epsilon \to 0$. Accordingly, by taking $\epsilon \to 0$, we have
\begin{equation}
    \lim_{\epsilon \to 0}  \frac{\beta(t) - \beta(t-\epsilon)}{\epsilon} \le - \frac{W_1(0, t^* - t)}{W(0, t^* - t)} (1 - \beta(t)).
\end{equation}
Similarly, by evaluating
\begin{equation}
    W(0, t^* - t) (1 -\beta(t)) \ge W(-\epsilon, t^* - t) (1 -\beta(t+\epsilon)),
\end{equation}
we obtain
\begin{equation}
    \lim_{\epsilon \to 0}  \frac{\beta(t+\epsilon) - \beta(t)}{\epsilon} \ge - \frac{W_1(0, t^* - t)}{W(0, t^* - t)} (1 - \beta(t)).
\end{equation}
Accordingly, $\beta$ is differentiable in $\mathbb{R}_+\setminus\{t^* - K\}$ and satisfies
\begin{equation}\label{eq: beta differential equation Model 2}
    \beta'(t) = - \frac{W_1(0, t^* - t)}{W(0, t^* - t)} (1 - \beta(t)).
\end{equation}
Since $\beta$ is a continuous function (by part (ii) of Theorem~\ref{thm: basic property of eqm Model 2}), we obtain $\beta$ as the unique solution to the differential equation \eqref{eq: beta differential equation Model 2} where  the initial condition satisfies $\beta(0) = 0$ (which follows from part (iii) of Theorem~\ref{thm: basic property of eqm Model 2}):
\begin{equation}
    \beta(t) = \beta^E(t, t^*) \coloneqq  1 - e^{\int_0^t \frac{W_1(0, t^* - \tau)}{W(0, t^* - \tau)}d\tau}.\tag{\ref{eq: eqm bid function one-shot always operating}}
\end{equation}
\end{proof}

\end{proof} 

\subsection{Proof of Theorem~\ref{thm: EO model beta is increasing in threshold}}

\begin{proof}
By direct calculation, we can show that
\begin{equation}
    W_1(s, t^* - t) = \begin{cases}
        - \eta\lambda e^{- \eta \lambda (K - s)} & \text{ if } t - s + K < t^*\\
        - \lambda e^{- \eta \lambda (t^* - t) - \lambda (K - s + t - t^*)} & \text{ if } t < t^* < t - s + K\\
        - \lambda e^{- \lambda (K - s)} & \text{ if } t > t^*,
    \end{cases}
\end{equation}
\begin{equation}
    W_1(s, t^* - t)/W(s, t^* - t) = 
    \begin{cases}
        -\dfrac{\eta\lambda e^{- \eta \lambda (K - s)}}{1 - e^{- \eta \lambda (K - s)}} & \text{ if } t - s + K < t^*\\
        -\dfrac{\lambda e^{- \eta \lambda (t^* - t) - \lambda (K - s + t - t^*)}}{1 - e^{- \eta \lambda (t^* - t) - \lambda (K - s + t - t^*)}} & \text{ if } t < t^* < t - s + K\\
        -\dfrac{\lambda e^{- \lambda (K - s)}}{1 - e^{- \lambda (K - s)}} & \text{ if } t > t^*,   
    \end{cases}
\end{equation}
\begin{equation}\label{eq: W2 divided by W}
    W_1(0, t^* - t)/W(0, t^* - t) = 
    \begin{cases}
        -\dfrac{\eta\lambda e^{- \eta \lambda K}}{1 - e^{- \eta \lambda K}} & \text{ if } t + K < t^*\\
        -\dfrac{\lambda e^{- \eta \lambda (t^* - t) - \lambda (K + t - t^*)}}{1 - e^{- \eta \lambda (t^* - t) - \lambda (K + t - t^*)}} & \text{ if } t < t^* < t + K\\
        -\dfrac{\lambda e^{- \lambda K}}{1 - e^{- \lambda K}} & \text{ if } t > t^*.   
    \end{cases}
\end{equation}
By direct calculation, we can confirm that \eqref{eq: W2 divided by W} is nonincreasing in $t^*$. 

It follows from
\begin{equation}
    \beta^E(t, t^*) \coloneqq 1 - e^{\int_0^t \frac{W_1(0, t^* - \tau)}{W(0, t^* - \tau)}d\tau},
    \end{equation}
that $\beta^E(t, t^*)$ is nondecreasing in $t^*$ given any $t$.
\end{proof} 

\subsection{Proof of Lemma~\ref{lem: tE in 0 K}}
\begin{proof} 
It follows from $\lambda y \le c < \lambda (K + y)$ and Theorem~\ref{thm: equilibrium closed form model 2} that $t^E \in [0, + \infty)$. Towards a contradiction, suppose that $t^E \in [K, + \infty)$. By Theorem~\ref{thm: basic property of eqm Model 2},   $\beta^E(0) = 0$, $\beta^E$ is strictly increasing, and $\beta^E(t) \in [0, 1)$. Since the  time-$0$ user bids $\beta^E(0) = 0$, his expected payoff is zero. However, by bidding $\beta(t^E - K/2)$ instead, 
he obtains a positive expected payoff. Accordingly, $(\beta^E, t^E)$ cannot be an equilibrium. This is a contradiction.

\end{proof} 

\subsection{Proof of Theorem~\ref{thm: model 2 beta function limit}}\label{subsec: model 2 beta function limit}

\begin{proof} 
When $\eta \to 0$, we have
\begin{equation}
    \frac{W_1(0, t^* - t)}{W(0, t^* - t)} = \begin{cases}
        - \dfrac{\lambda e^{- \lambda(K + t - t^*)}}{1 - e^{-\lambda (K + t- t^*)}} & \text{ if } t \le t^*;\\
        - \dfrac{\lambda e^{\lambda K}}{1 - e^{- \lambda K}} & \text{ otherwise.}
    \end{cases}
\end{equation}
Recall that 
\begin{equation} 
    \beta^E(t, t^*) \coloneqq 1 - e^{\int_0^t \frac{W(0, t^* - \tau)}{W(0, t^* - \tau)}d\tau}. \tag{\ref{eq: eqm bid function model 2}, revisited}
\end{equation}
For $t < t^*$, we have
\begin{equation}
    \frac{W_1(0, t^* - t)}{W(0, t^* - t)} = - \left(\log\left(1 - e^{-\lambda(K + t - t^*)} \right) \right)'.
\end{equation}
Thus,
\begin{align}
    \beta^E(t, t^*) & = 1 - e^{- \int_0^t \left(\log\left(1 - e^{-\lambda(K + \tau - t^*)} \right) \right)'d\tau}\\
    & = 1 - \frac{1 - e^{-\lambda(K - t^*)}}{1 - e^{-\lambda(K + t - t^*)}}.
\end{align}

For $t \ge t^*$, we have
\begin{equation}
    \frac{W_1(0, t^* - t)}{W(0, t^* - t)} = - \left(\log\left(1 - e^{-\lambda K} \right) \right)'.
\end{equation}
Recall that this is independent of $t$. Thus, 
\begin{align}
    \beta^E(t, t^*) &= 1 - e^{\int_0^{t^*} \frac{W_2(\tau, 0; t^*)}{W(\tau, 0; t^*)}d\tau \int_{t^*}^{t} \frac{W_2(\tau, 0; t^*)}{W(\tau, 0; t^*)}d\tau},\\
    &= 1 - \dfrac{1 - e^{-\lambda(K - t^*)}}{1 - e^{-\lambda K}} e^{- \frac{\lambda e^{- \lambda K}}{1 - e^{- \lambda K}}(t -t^*)},
\end{align}
as desired.
\end{proof} 

\subsection{Proof of Lemma~\ref{lem: beta increasing in tstar}}\label{subsec: beta increasing in tstar}

\begin{proof}

To show that $\beta^E(t, t^*)$ is strictly convex in $t^*$, we compute the partial derivative of $\beta^E$ as follows:
\begin{equation}
    \beta^E_2(t, t^*) = \frac{\lambda e^{-\lambda (K - t^*)}(1 - e^{-\lambda t})}{(1 - e^{- \lambda (K + t - t^*)})^2}.
\end{equation}
The numerator is increasing in $t^*$, and the denominator is decreasing in $t^*$. Thus, $\beta_2^E$ is increasing in $t^*$. Accordingly, $\beta^E$ is strictly convex in $t^*$.
\end{proof}

\subsection{Proof of Theorem~\ref{thm: comparative statics M}} \label{subsec: comparative statics M}

\begin{proof} 
    \textbf{(Monotonicity)}
    Recall that $M^*(t^*) = \int_{[t^* - K]^+}^{t^*} \beta^E(t, t^*) dt$.
    By Lemma~\ref{lem: beta increasing in tstar}, $\beta^E(t, t^*)$ is nonnegative and nondecreasing in $t^*$. Accordingly, $M^*$ is increasing.

    \noindent
    \textbf{(Convexity given $t^* \le K$ and $\eta = 0$)}
    For $t^* \le K$, we have
    \begin{equation}
        \frac{d}{dt^*}M^*(t^*) = \beta^E(t^*, t^*)  + \int_{0}^{t^*}\beta_2^E(t, t^*)dt.
    \end{equation}
    $\beta^E$ is increasing in both arguments. Furthermore, since $\beta^E$ is increasing in $t^*$, $\beta_2^E(t, t^*) \ge 0$. Moreover, by Lemma~\ref{lem: beta increasing in tstar}, $\beta^E$ is strictly convex in $t^*$ for $\eta = 0$, and thus $\beta_2^E(t, t^*)$ is increasing in $t^*$. Accordingly, the first-order derivative of $M^*$ is increasing, implying that $M^*$ is strictly convex for $t^* \le K$ and $\eta = 0$.
\end{proof} 

\subsection{A Closed-form Representation of $M^*$}\label{subsec: M closed form limit}

Assuming $\eta = 0$, we compute the closed form of the $M^*$ function for $t^* \in [0, K]$. First, by Theorem~\ref{thm: model 2 beta function limit} and the definition of $M$, we obtain
\begin{align}
    M(t^*; \beta^E(\cdot; t^*)) &= \int_0^{t^*} \beta^E(t, t^*) dt \\
    & = \int_0^{t^*} \frac{e^{- \lambda (K - t^*)}}{1 - e^{- \lambda (K + t - t^*)}} - \frac{e^{- \lambda (K + t - t^*)}}{1 - e^{- \lambda (K + t - t^*)}}dt. \label{eq: M star calculation}
\end{align}

For the integration of the former part, we use the fact that
\begin{equation}
    \int \frac{a}{1 - e^{bx + c}}dx = a \left(x - \frac{\log (1 - e^{bx + c}))}{b} \right) + C.
\end{equation}
Using this, we obtain
\begin{align}
    &\int_0^{t^*} \frac{e^{- \lambda (K - t^*)}}{1 - e^{- \lambda (K + t - t^*)}}dt\\
    & = \left[e^{-\lambda(K-t^*)} \left(t + \frac{\log \left( 1 - e^{- \lambda(K + t - t^*)} \right)}{\lambda} \right)  \right]_0^{t^*}\\
    & = e^{-\lambda(K-t^*)}\left\{t^* + \frac{\log \left( 1 - e^{- \lambda K} \right)}{\lambda} - \frac{\log \left( 1 - e^{- \lambda(K - t^*)} \right)}{\lambda}  \right\}
\end{align}

By direct calculation, we can compute the latter term of \eqref{eq: M star calculation} as follows.
\begin{align}
    & \int_0^{t^*} \frac{e^{- \lambda (K + t - t^*)}}{1 - e^{- \lambda (K + t - t^*)}}dt\\
    & = - \frac{1}{\lambda}\left[\log\left(1 - e^{- \lambda(K + t - t^*)} \right) \right]_0^{t^*}\\
    & = - \frac{1}{\lambda}\left[\log\left(1 - e^{-\lambda K} \right) - \log \left(1 - e^{- \lambda(K - t^*)} \right) \right]
\end{align}

Using these results, we obtain the following, as desired.
\begin{equation}\label{eq: M closed form limit}
    M(t^*, \beta^E(\cdot; t^*)) = e^{-\lambda(K -t^*)} t^* + \frac{1}{\lambda}\left( 1 - e^{-\lambda(K - t^*)} \right)\left(\log \left(1 - e^{-\lambda (K - t^*)} \right) - \log \left(1 - e^{- \lambda K} \right) \right).
\end{equation}

\subsection{Proof of Proposition~\ref{prop: EO model tE response y and c}}
\begin{proof} 
    First, we prove the statement (i). For this proof, we represent the equilibrium threshold time given the block reward $y$ by $t^E(y)$. Take $y'$ and $y''$ that satisfy $t^E(y'), t^E(y'') > 0$. Take $\omega \in (0,1)$ arbitrarily, and let $y''' = \omega y' + (1-\omega) y''$. Then, we have the following:
    \begin{align}
        & M^*(\omega t^E(y') + (1 - \omega) t^E(y'')) \\
        & < \omega M^*(t^E(y')) + (1 - \omega) M^*(t^E(y'')) \quad\quad\text{ (by Theorem~\ref{thm: comparative statics M})}\\
        & = \frac{c}{\lambda} - (\omega y' + (1-\omega) y'')
        \quad\quad\text{ (by \eqref{eq: tE determination})}  
        \\
        & = \frac{c}{\lambda} - y'''\\
        &= M^*(t^E(y''')) \quad\quad\text{ (by \eqref{eq: tE determination})} .
    \end{align}
    Since $M^*(t)$ is increasing in $t$ by Theorem~\ref{thm: comparative statics M}, we have $t^E(y''') > \omega t^E(y') + (1 - \omega) t^E(y'')$, implying that $t^E$ is strictly concave in $y$ for $t^E > 0$. 

    The proof for the statement (ii) is similar.
\end{proof} 

\subsection{Proof of Lemma~\ref{lem: beta E is decreasing in K}} \label{subsec: beta E is decreasing in K}

\begin{proof} 
Recall 
\begin{equation}
    \beta^E(t, t^*) = \frac{e^{- \lambda (K - t^*)} - e^{- \lambda (K + t - t^*)}}{1 - e^{- \lambda (K + t - t^*)}}.
\end{equation}
The conclusion follows that the numerator is decreasing in $K$ and the denominator is increasing in $K$.
\end{proof} 

\subsection{Proof of Proposition~\ref{prop: tE is increasing in K}}
\begin{proof} 
Recall that $t^E$ solves
\begin{equation}
    \lambda (M^*(t^E; K) + y) = c.
\end{equation}
The conclusion follows from the fact that $M^*$ is increasing in $t^*$ and decreasing in $K$.
\end{proof} 

\subsection{Proof of Proposition~\ref{prop: stationary distribution endogenous operation model}}

\begin{proof}
The following equations characterize the stationary distribution, where we let $\dot{\psi}$ be the partial derivative of $\psi$ with respect to $t$.
\begin{itemize}
    \item For $t = 0$, $\psi(0; t^*) = \lambda (1 - \Psi(t^*; t^*))$. 
    \item For $t \in (0, t^*)$, $\psi(t; t^*) = \psi(0, t^*)$. 
    \item For $t \in [t^*, +\infty)$, $\dot{\psi}(t; t^*) = - \lambda \psi(t; t^*)$.
\end{itemize}
Solving this system, we obtain the following stationary distribution $\Psi$:
\begin{equation}
    \psi(t; t^*) = \begin{cases}
        \dfrac{\lambda t^*}{1 + \lambda t^*} & \text{ if } t \in [0, t^*) \vspace{0.5em}\\
        \dfrac{\lambda e^{-\lambda (t - t^*)}}{1 + \lambda t^*} & \text{ if } t \in [t^*, +\infty)
    \end{cases}.
\end{equation}
\end{proof} 

\subsection{Proof of Lemma~\ref{lem: beta is decreasing in lambda}}
\begin{proof} 
    Recall that when $\eta = 0$ and $t \le t^*$, $\beta^E(t, t^*; \lambda)$ is given by
    \begin{equation}
        \beta^E(t, t^*; \lambda) = 1 - \frac{1 - e^{-\lambda(K - t^*)}}{1 - e^{-\lambda(K + t - t^*)}}
    \end{equation}
    By taking a partial derivative with respect to $\lambda$, we have
    \begin{equation}
        \frac{\partial}{\partial \lambda}\beta^E(t, t^*; \lambda) = - \frac{(K - t^*)e^{- \lambda (K - t^*)} (1 - e^{- \lambda t}) - t e^{-\lambda (K - t^* + t)}(1 - e^{- \lambda (K - t^*)})}{(1 - e^{-\lambda(K + t - t^*)})^2}.
    \end{equation}
    Thus, it suffices to show that
    \begin{equation}\label{eq: betaE decreasing in lambda proof 1}
        (K - t^*) (1 - e^{- \lambda t}) - t e^{-\lambda t}(1 - e^{- \lambda (K - t^*)}) > 0.
    \end{equation}
    Let $a = K - t^*$ and $b = t$. Then, we have $a \in (0, K]$ and $b \in [0, t^*]$, and \eqref{eq: betaE decreasing in lambda proof 1} can be rewritten as
    \begin{equation}\label{eq: betaE decreasing in lambda proof 2}
        a  (1 - e^{- \lambda b}) - b e^{-\lambda b}(1 - e^{- \lambda a}) >  0.
    \end{equation}
    If we substitute $a = 0$ into \eqref{eq: betaE decreasing in lambda proof 2}, the left-hand side of \eqref{eq: betaE decreasing in lambda proof 2} is equal to zero. We show that the left-hand side of \eqref{eq: betaE decreasing in lambda proof 2} is increasing in $a$. Its derivative with respect to $a$ is 
    \begin{equation}
        1 - e^{- \lambda b} - b\lambda e^{- \lambda (a + b)}.
    \end{equation}
    Since this is increasing in $a$, it suffices to show that
    \begin{equation}\label{eq: betaE decreasing in lambda proof 3}
        f(b) \coloneqq 1 - (1 - \lambda b) e^{- \lambda b} \geq 0
    \end{equation}
    for $b \in [0, t^*]$.
    The derivative of the left-hand side of \eqref{eq: betaE decreasing in lambda proof 3} with respect to $b$ is
    \begin{equation}
        e^{- \lambda b}\lambda (b \lambda - 2),
    \end{equation}
    implying that the left-hand side of \eqref{eq: betaE decreasing in lambda proof 3} is increasing in $b$ for $b < 2/\lambda$ and decreasing in $b$ for $b > 2/\lambda$. Accordingly, $f(b) > \min\{f(0), \lim_{b \to \infty}f(b)\}$ for all $b \in (0, +\infty) \supset (0, t^*]$. Finally, by direct calculation, we can verify that $f(0) = 0$ and $\lim_{b \to \infty}f(b) = 0$.
\end{proof} 

\subsection{Proof of Theorem~\ref{thm: tE is nonmonotonic in lambda}}
\begin{proof} 
When $\lambda < c/K$, we have
\begin{equation}
    \lambda M^*(t^*; \lambda) < \lambda \lim_{t^* \to K} M^*(t^*; \lambda) = \lambda K < c/K.
\end{equation}
Together with $y = 0$, we have $t^E = + \infty$, that is, the miner never works.

By contrast, for any $\lambda > c/K$, we have
\begin{equation}
    \lambda \lim_{t^* \to K}M^*(t^*; \lambda) = \lambda K> c.
\end{equation}
Since $M^*$ is continuous and increasing in $t^*$, we have either (i) $\lambda M^*(0; \lambda) > c$ or (ii) there exists $t^E \in (0, K)$ such that $\lambda M^*(t^E, \lambda) = c$. Accordingly, $t^E \in [0, K)\cup\{-\infty\}$.

To prove $t^E \to K$ as $\lambda \to \infty$, we first present the following lemma.

\begin{lem}\label{lem: tE nonmonotonic in lambda}
    $\lambda M^*(t^*; \lambda) \to 0$ as $\lambda \to \infty$ for any fixed $t^*$. 
\end{lem}

\begin{proof} 
Recall that, for $t^* \in [0, K)$,
\begin{equation}
    \lambda M^*(t^*; \lambda) = \lambda \int_0^{t^*} \beta(t, t^*) dt = \lambda\int_0^{t^*} \left(1 - \dfrac{1 - e^{-\lambda(K - t^*)}}{1 - e^{-\lambda(K + t - t^*)}}\right)dt.
\end{equation}
Hence, it suffices to show that
\begin{equation}
    \lambda \left(1 - \dfrac{1 - e^{-\lambda(K - t^*)}}{1 - e^{-\lambda(K + t - t^*)}}\right) = \lambda \cdot \dfrac{e^{-\lambda(K + t - t^*)} - e^{-\lambda(K + t - t^*)}}{1 - e^{-\lambda(K + t - t^*)}} \to 0 \text{ as } \lambda \to \infty.
\end{equation}
To see this,
\begin{align}
    \lim_{\lambda \to \infty}\lambda \cdot \dfrac{e^{-\lambda(K- t^*)} - e^{-\lambda(K + t - t^*)}}{1 - e^{-\lambda(K + t - t^*)}} 
    &= \lim_{\lambda \to \infty}\lambda \cdot \dfrac{e^{-\lambda(K - t^*)}}{1 - e^{-\lambda(K + t - t^*)}} - \lim_{\lambda \to \infty} \lambda \cdot \dfrac{e^{-\lambda(K - t^*)}}{1 - e^{-\lambda(K + t - t^*)}}\\
    &= \lim_{\lambda \to \infty} \dfrac{1}{(K - t^*) e^{\lambda (K + t - t^*)}} - \lim_{\lambda \to \infty} \dfrac{1}{(K - t^*) e^{\lambda (K + t - t^*)}}\\
    &= 0,
\end{align}
where we used l'H\^opital's law for the second equality.
\end{proof} 

Now, we will show $t^E \to K$ as $\lambda \to \infty$. Take $\epsilon > 0$ arbitrarily. By Lemma~\ref{lem: tE nonmonotonic in lambda}, there exists $\bar{\lambda} <\infty$ such that for any $\lambda > \bar{\lambda}$, $\lambda M^*(K - \epsilon; \lambda) < c$. Since $M^*$ is increasing in $t^*$, for all $t^* < K - \epsilon$, we also have $\lambda M^*(t^*; \lambda) < c$. However, as we have shown, for any $\lambda$, there exists $t^E(\lambda) < K$ that satisfies $\lambda M^*(t^E(\lambda); \lambda) = c$. Accordingly, for all $\epsilon > 0$, there exists $\bar{\lambda} <\infty$ such that for all $\lambda > 0$, we have $t^E(\lambda) \in (K - \epsilon, K)$, as desired.
\end{proof} 

\subsection{Proof of Proposition~\ref{prop: bar SW and tO}}

\begin{proof}
Given $\lambda K > c$, $SW(K) > SW(t^*)$ holds for all $t^* \in (K, +\infty) \cup \{+\infty\}$ because for all $t \ge K$, net flow surplus from operation is $\lambda K - c > 0$.

Given $t^* \in [0, K]$, $SW(t^*)$ can be expressed as follows.
\begin{align}
    SW(t^*) := & \int_0^\infty (W(0, t^* - t; t^*) - c \textbf{1}\left\{t \ge t^*\right\}) \psi(t; t^*) dt\\
    & = \frac{1}{1 + \lambda t^*}\left[\lambda \int_0^{t^*}(1 - e^{-\lambda(K - (t^* - t))}) dt + \int_{t^*}^\infty (1 - e^{-\lambda K})\lambda e^{-\lambda (t - t^*)} - c \right] \\
    & =  1 - \frac{e^{- \lambda(K - t^*)} + c}{1 + \lambda t^*}.
\end{align}
The derivative of $SW$ with respect to $t^*$ is
\begin{equation}
    SW'(t^*) = \frac{- \lambda^2 t^* e^{-\lambda (K - t^*)} + \lambda c}{\left(1 + \lambda t^*\right)^2}.
\end{equation}
Its denominator is always positive and its numerator is strictly decreasing in $t^*$. Furthermore, $SW'(0) = \lambda c > 0$ and $SW'(K) = - \lambda (\lambda K - c)/(1 + \lambda K)^2 < 0$, implying that neither $t^* = 0$ nor $K$ is optimal. Accordingly, there is a unique threshold time $t^O \in (0, K)$ that maximizes $SW$. The efficient threshold time $t^O$ solves the following first-order condition.
\begin{equation}
    t^* e^{-\lambda (K - t^*)} = \frac{c}{\lambda}.\tag{\ref{eq: efficient t star}, revisited}
\end{equation}
\end{proof}

\subsection{Prof of Lemma~\ref{lem: to tstar comparison}}
\begin{proof} 
For $t < t^*$,
\begin{align}
    \beta^E(t, t^*) &= 1 - \frac{1 - e^{-\lambda (K - t^*)}}{1 - e^{- \lambda(K + t - t^*)}} \\
    &\le \beta^E(t^*, t^*) \\
    & = 1 - \frac{1 - e^{-\lambda (K - t^*)}}{1 - e^{- \lambda K}}\\
    & = e^{-\lambda K} \frac{e^{\lambda t^*} - 1}{1 - e^{-\lambda K}}.
\end{align}
Accordingly, it suffices to show that
\begin{equation}
    \frac{e^{\lambda t^*} - 1}{1 - e^{-\lambda K}} < e^{-\lambda (K - t^*)},
\end{equation}
or equivalently,
\begin{equation}
    e^{\lambda t^*} - 1 < - e^{-\lambda (K - t^*)} - e^{\lambda t^*}
\end{equation}
or
\begin{equation}
    1 > e^{-\lambda (K - t^*)}.
\end{equation}
This inequality indeed holds as $t^* \in [0, K)$.
\end{proof} 

\subsection{Proof of Theorem~\ref{thm: differential equation W-tilde satisfies}}
\begin{proof}
We first show that $\tilde{W}$ is continuous. To bound $\tilde{W}(s)$, we consider the following two processes:
\begin{enumerate}
    \item[(i)] For a time interval with length $\Delta > 0$, no block arrives. After that, the process returns to the original one.
    \item[(ii)] For a time interval with length $\Delta > 0$, the user's transaction request is validated immediately if a block arrives. After that, the process returns to the original one.
\end{enumerate}
Clearly, the original process validates the user's request earlier than Case (i) and later than Case (ii). Accordingly, we have
\begin{equation}
    \tilde{W}(s) \ge e^{-(\rho + \lambda) \Delta} \tilde{W}(s + \Delta)
\end{equation}
and
\begin{align}
    \tilde{W(}s) & \le \int_0^\Delta \lambda e^{-(\rho + \lambda)u} du + e^{-(\rho + \lambda) \Delta} \tilde{W}(s + \Delta)\\
    & < \lambda \Delta +  e^{-(\rho + \lambda) \Delta} \tilde{W}(s + \Delta).
\end{align}
Taking the limit of $\Delta \to 0$, we have $\lim_{\Delta \downarrow 0}\tilde{W}(s + \Delta) = \tilde{W}(s)$. Furthermore, starting from $s - \Delta$, we can apply the same argument. Then, we have $\lim_{\Delta \downarrow 0}\tilde{W}(s - \Delta) = \tilde{W}(s)$. Accordingly, $\tilde{W}(s + \Delta) \to s$ as $\Delta \to 0$, and therefore, $\tilde{W}$ is continuous.

Since the block arrival follows a Poisson process with intensity $\lambda$, for any $\Delta > 0$,  we have
\begin{equation}\label{eq: W function recursive}
    \tilde{W}(s) = \int_0^\Delta \lambda e^{-(\rho + \lambda) u} \tilde{W}^*(s + u - K) du + e^{-(\rho + \lambda) \Delta} \tilde{W}(s + \Delta).
\end{equation}
Note that, the continuity and boundedness of $\tilde{W}$ ensures that $\tilde{W}^*$ is integrable. \eqref{eq: W function recursive} can be rewritten as
\begin{equation}\label{eq: W derivative derivation}
    \frac{\tilde{W}(s+\Delta) - \tilde{W}(s)}{\Delta} = \frac{1 - e^{-(\rho + \lambda)\Delta}}{\Delta} \tilde{W}(s + \Delta ) - \frac{1}{\Delta} \int_0^\Delta \lambda e^{-(\rho+ \lambda) u}\tilde{W}^*(s+ u - K)du.
\end{equation}
Since $\tilde{W}$ is continuous, for all $s \neq 0$, $\tilde{W}^*(s)$ is also continuous, whereas $\tilde{W}^*(0) = \tilde{W}(0) < 1$ implies that $\tilde{W}^*$ is discontinuous at $s = 0$. For all $s-K \neq 0$, taking the limit of $\Delta \to 0$, we obtain
\begin{equation}\label{eq: W diff}
    \tilde{W}'(s) = (\rho + \lambda) \tilde{W}(s) - \lambda \tilde{W}^*(s - K).
\end{equation}
By the Picard-Lindel{\"o}f theorem, for any initial condition, \eqref{eq: W diff} has a unique solution.

To prove $\lim_{s \to - \infty}\tilde{W}(s) = \lambda/(\rho + \lambda)$ and $\lim_{s \to + \infty}\tilde{W}(s) = 0$, we use the following lemma.

\begin{lem}\label{lem: W tilde upper bound}
    In the EO model, for all $s \in \mathbb{R}$, $\tilde{W}(s) \le \lambda/(\rho + \lambda)$. Furthermore, for all $n \in \mathbb{N}$, for any $s \ge nK$, $\tilde{W}(s) \le (\lambda/(\rho + \lambda))^{n + 1}$.
\end{lem}

\begin{proof}
    For all $s \in \mathbb{R}$,
    \begin{align}
        \tilde{W}(s) & = \int_0^\infty (\textbf{1}\{s + t < K\} + \textbf{1}\{s + t \ge K\}\tilde{W}(s - K)) \lambda e^{-(\rho + \lambda)t}dt\\
        & \le \int_0^\infty \lambda e^{-(\rho + \lambda)t}dt\\
        & = \dfrac{\lambda}{\rho + \lambda}.
    \end{align}
    Suppose that $\tilde{W}(s) \le (\lambda/(\rho + \lambda))^{n' + 1}$ for any $s \ge n'K$. Then, for any $s \ge (n' + 1) K$,
    \begin{align}
        \tilde{W}(s) & = \int_0^\infty \tilde{W}(s - K)) \lambda e^{-(\rho + \lambda)t}dt\\
        & \le \int_0^\infty \left(\dfrac{\lambda}{\rho + \lambda}\right)^{n' + 1} \lambda e^{-(\rho + \lambda)t}dt\\
        & = \left(\dfrac{\lambda}{\rho + \lambda}\right)^{n' + 2}.
    \end{align}
    By mathematical induction, the proposition holds for all $n \in \mathbb{N}$.
\end{proof}

By Lemma~\ref{lem: W tilde upper bound}, $\tilde{W}(s) \le \lambda/(\rho + \lambda)$ for all $s$. To evaluate $\tilde{W}(s)$ from below, we consider the following alternative process: If a block arrives between now and $s$ unit of time later, then the user's transaction request is validated. After that, it is never validated. Clearly, the original process validates the user's request earlier than the alternative process. Thus,
\begin{align}
    \tilde{W}(-s) &> \int_0^s \lambda e^{-(\rho + \lambda) t}dt\\
    & = \frac{\lambda}{\rho + \lambda} - \frac{\lambda}{\rho + \lambda} e^{-(\rho + \lambda) s} \to \frac{\lambda}{\rho + \lambda} \text{ as } s \to \infty.
\end{align}
Thus, $\lim_{s \to - \infty}\tilde{W}(s) = \lambda/(\rho + \lambda)$.

Finally, $\tilde{W}(s) > 0$ for all $s$ and $\tilde{W}(s) \le (\lambda/(\rho + \lambda))^{n + 1}$ for all $n \in \mathbb{N}$ for any $s \ge nK$ immediately imply that $\lim_{s \to + \infty}\tilde{W}(s) = 0$.
\end{proof} 

\end{document}

%% file: preamble.tex
\usepackage{setspace}
\usepackage[utf8]{inputenc}
\usepackage[driver=dvipdfm,margin=1in]{geometry}

\usepackage[hyphens]{url}

\usepackage{amsmath}
\usepackage{mathtools}
\usepackage[colorlinks,urlcolor=blue, citecolor=blue, menucolor=blue, hypertexnames=false]{hyperref}
\usepackage{cleveref}
\usepackage{autonum}

\crefname{equation}{Eq.}{Eq.}
\crefname{figure}{Figure}{Figures}
\crefname{table}{Table}{Tables}
\crefname{thm}{Theorem}{Theorems}
\crefname{lem}{Lemma}{Lemmas}
\crefname{prop}{Proposition}{Propositions}

\usepackage{amsthm}
\usepackage{amssymb}
\usepackage{booktabs}
\usepackage{mathrsfs}

\makeatletter

\usepackage{amssymb}
\usepackage{amsthm}
\usepackage{bbm}
\usepackage{graphics}
\usepackage[dvipdfmx]{graphicx}
\usepackage{accents}
\usepackage{placeins}
\usepackage{enumerate}
\usepackage{float}
\usepackage{caption}
\usepackage{subcaption}
\usepackage{ragged2e}
\usepackage[longnamesfirst]{natbib}
\usepackage{enumerate}
\usepackage{comment}
\usepackage{ascmac}
\usepackage{mathtools}
\usepackage{empheq}
\usepackage{empheq}
\usepackage[dvipsnames]{xcolor}

\newtheorem{thm}{Theorem}
\newtheorem{lem}{Lemma}
\newtheorem{prop}{Proposition}

\theoremstyle{remark}
\newtheorem{remark}{Remark}

\theoremstyle{definition}
\newtheorem{defn}{Definition}

\usepackage{amssymb}
\usepackage{bm}
\usepackage{ascmac}
\usepackage{setspace}
\usepackage[at]{easylist}

\theoremstyle{plain}

\providecommand{\keywords}[1]
{
  \small	
  \textbf{Keywords:} #1
}
\providecommand{\JEL}[1]
{
  \small	
  \textbf{JEL Codes:} #1
}

\usepackage{algorithm,algpseudocode}

\DeclareMathOperator*{\argmax}{arg\,max}

\makeatother

\makeatletter
\let\save@mathaccent\mathaccent
\newcommand*\if@single[3]{%
  \setbox0\hbox{${\mathaccent"0362{#1}}^H$}%
  \setbox2\hbox{${\mathaccent"0362{\kern0pt#1}}^H$}%
  \ifdim\ht0=\ht2 #3\else #2\fi
  }
\newcommand*\rel@kern[1]{\kern#1\dimexpr\macc@kerna}
\newcommand*\widebar[1]{\@ifnextchar^{{\wide@bar{#1}{0}}}{\wide@bar{#1}{1}}}
\newcommand*\wide@bar[2]{\if@single{#1}{\wide@bar@{#1}{#2}{1}}{\wide@bar@{#1}{#2}{2}}}
\newcommand*\wide@bar@[3]{%
  \begingroup
  \def\mathaccent##1##2{%
    \let\mathaccent\save@mathaccent
    \if#32 \let\macc@nucleus\first@char \fi
    \setbox\z@\hbox{$\macc@style{\macc@nucleus}_{}$}%
    \setbox\tw@\hbox{$\macc@style{\macc@nucleus}{}_{}$}%
    \dimen@\wd\tw@
    \advance\dimen@-\wd\z@
    \divide\dimen@ 3
    \@tempdima\wd\tw@
    \advance\@tempdima-\scriptspace
    \divide\@tempdima 10
    \advance\dimen@-\@tempdima
    \ifdim\dimen@>\z@ \dimen@0pt\fi
    \rel@kern{0.6}\kern-\dimen@
    \if#31
      \overline{\rel@kern{-0.6}\kern\dimen@\macc@nucleus\rel@kern{0.4}\kern\dimen@}%
      \advance\dimen@0.4\dimexpr\macc@kerna
      \let\final@kern#2%
      \ifdim\dimen@<\z@ \let\final@kern1\fi
      \if\final@kern1 \kern-\dimen@\fi
    \else
      \overline{\rel@kern{-0.6}\kern\dimen@#1}%
    \fi
  }%
  \macc@depth\@ne
  \let\math@bgroup\@empty \let\math@egroup\macc@set@skewchar
  \mathsurround\z@ \frozen@everymath{\mathgroup\macc@group\relax}%
  \macc@set@skewchar\relax
  \let\mathaccentV\macc@nested@a
  \if#31
    \macc@nested@a\relax111{#1}%
  \else
    \def\gobble@till@marker##1\endmarker{}%
    \futurelet\first@char\gobble@till@marker#1\endmarker
    \ifcat\noexpand\first@char A\else
      \def\first@char{}%
    \fi
    \macc@nested@a\relax111{\first@char}%
  \fi
  \endgroup
}
\makeatother